\documentclass[12pt,onecolumn,twoside]{IEEEtran}

\usepackage{calc}

\usepackage{multirow}
\usepackage{cite}
\usepackage{graphicx,subfigure}
\usepackage{psfrag}
\usepackage{amsmath,amssymb,amsthm}
\usepackage{color}
\usepackage{tikz} 

\usepackage{cases}
\usepackage{algorithm}
\usepackage[noend]{algpseudocode}

\usepackage{mathtools}
\usepackage[utf8]{inputenc}
\usepackage{textcomp}
\usepackage{pgfplots}
\pgfplotstableread[row sep=crcr]{
2 10\\
4 8\\
6 6\\
8 6\\
10 5\\
12 5\\
14 5\\
16 5\\
18 5\\
20 5\\
}\mytable

\DeclarePairedDelimiter\ceil{\lceil}{\rceil}
\DeclarePairedDelimiter\floor{\lfloor}{\rfloor}
\makeatletter
\def\BState{\State\hskip-\ALG@thistlm}
\makeatother

\DeclareMathOperator*{\defeq}{\triangleq}

\newtheorem{theorem}{Theorem}

\newtheorem{lemma}{Lemma}
\newtheorem{definition}{Definition}
\newtheorem{proposition}{Proposition}

\newcommand{\bit}{\begin{itemize}}
\newcommand{\eit}{\end{itemize}}

\newcommand{\bc}{\begin{center}}
\newcommand{\ec}{\end{center}}

\newcommand{\ba}{\begin{array}}
\newcommand{\ea}{\end{array}}

\newcommand{\beq}{\begin{equation}}
\newcommand{\eeq}{\end{equation}}

\newcommand{\beqn}{\begin{equation*}}
\newcommand{\eeqn}{\end{equation*}}

\newcommand{\bean}{\begin{eqnarray*}}
\newcommand{\eean}{\end{eqnarray*}}
\newcommand{\bea}{\begin{eqnarray}}
\newcommand{\eea}{\end{eqnarray}}

\def\F{\mathbb{F}}

\def\hv{\boldsymbol{h}}

\def\sv{\boldsymbol{s}}

\def\uv{\boldsymbol{u}}

\def\yv{\boldsymbol{y}}
\def\zv{\boldsymbol{z}}

\newcommand{\Ac}{{\mathcal A}}
\newcommand{\Bc}{{\mathcal B}}
\newcommand{\Cc}{{\mathcal C}}

\newcommand{\Mc}{{\mathcal M}}

\newcommand{\Sc}{{\mathcal S}}

\newcommand{\Wc}{{\mathcal W}}

\newcommand{\T}{{\scriptscriptstyle\mathsf{T}}}

\newcommand{\non}{\nonumber}

\newcommand{\Hen}{\mathbb{H}}

\newcommand{\Mtn}{M_{\text{total}}}
\newcommand{\Mwtn}{M_{\text{worst}}}
\newcommand{\Mt}{M_{\text{total}}^{*}}
\newcommand{\Mwt}{M_{\text{worst}}^{*}}

\newcommand{\Rp}{\gamma_{r,p}}
\newcommand{\Cp}{\gamma_{c,p}}

\newcommand{\ksum}{(k+1)^\dagger}
\newcommand{\ksumm}{(k+2)^\dagger}
\algnewcommand\algorithmicforeach{\textbf{for each}}
\algnewcommand{\IfThenElse}[3]{
  \State \algorithmicif\ #1\ \algorithmicthen\ #2\ \algorithmicelse\ #3}
\algdef{S}[FOR]{ForEach}[1]{\algorithmicforeach\ #1\ \algorithmicdo}
\renewcommand\algorithmicdo{}
\renewcommand\algorithmicthen{}

\usepackage{mathtools}
\usepackage{xifthen}
\usepackage{booktabs}
\newcommand{\bunderline}[1]{\underline{#1\mkern-4mu}\mkern4mu }

\newcommand{\hvu}{\bunderline{\boldsymbol{h}}}
\newcommand{\xvu}{\bunderline{\boldsymbol{x}}}
\newcommand{\uvu}{\bunderline{\boldsymbol{u}}}
\newcommand{\vvu}{\bunderline{\boldsymbol{v}}}
\newcommand{\yvu}{\bunderline{\boldsymbol{y}}}
\newcommand{\zvu}{\bunderline{\boldsymbol{z}}}
\newcommand{\svu}{\bunderline{\boldsymbol{s}}}
\newcommand{\qvu}{\bunderline{\boldsymbol{q}}}

\newcommand{\fvu}{\bunderline{\boldsymbol{f}}}

\newcommand{\xvut}[1][]{\ifthenelse{\isempty{#1}}{\xvu_{t}}{\xvu_{#1}}}
\newcommand{\hvut}[1][]{\ifthenelse{\isempty{#1}}{\hvu_{t}}{\hvu_{#1}}}
\newcommand{\uvut}[1][]{\ifthenelse{\isempty{#1}}{\uvu_{t}}{\uvu_{#1}}}
\newcommand{\vvut}[1][]{\ifthenelse{\isempty{#1}}{\vvu_{t}}{\vvu_{#1}}}
\newcommand{\svut}[1][]{\ifthenelse{\isempty{#1}}{\svu_{t}}{\svu_{#1}}}
\newcommand{\qvut}[1][]{\ifthenelse{\isempty{#1}}{\qvu_{t}}{\qvu_{#1}}}
\newcommand{\zvut}[1][]{\ifthenelse{\isempty{#1}}{\zvu_{t}}{\zvu_{#1}}}
\newcommand{\yvut}[1][]{\ifthenelse{\isempty{#1}}{\yvu_{t}}{\yvu_{#1}}}
\newcommand{\fvut}[1][]{\ifthenelse{\isempty{#1}}{\fvu_{t}}{\fvu_{#1}}}

\newcommand{\hvt}[1][]{\ifthenelse{\isempty{#1}}{\hv_{t,m}}{\hv_{t,#1}}}
\newcommand{\hvot}[1][]{\ifthenelse{\isempty{#1}}{\hv_{1,m}}{\hv_{1,#1}}}
\newcommand{\zt}[1][]{\ifthenelse{\isempty{#1}}{\zv_{t}}{\zv_{#1}}}
\newcommand{\yt}[1][]{\ifthenelse{\isempty{#1}}{\yv_{t}}{\yv_{#1}}}
\newcommand{\st}[1][]{\ifthenelse{\isempty{#1}}{\sv_{t}}{\sv_{#1}}}

\begin{document}
\sloppy

\title{Distributed Computing with Heterogeneous Communication Constraints:  The Worst-Case Computation Load and Proof by Contradiction}

\author{Nishant Shakya, Fan Li and Jinyuan Chen
\thanks{Nishant Shakya, Fan Li and Jinyuan Chen are with Louisiana Tech University, Department of Electrical Engineering, Ruston, LA 71272, US (emails: nsh018@latech.edu, fli005@latech.edu, jinyuan@latech.edu).  The work was partly supported by Louisiana Board of Regents Support Fund (BoRSF) Research Competitiveness Subprogram (RCS) under grant 32-4121-40336. This work was presented in part at the 52nd Annual Asilomar Conference on Signals, Systems, and Computers, October 2018.}}

\maketitle
\pagestyle{headings}

\begin{abstract}

We consider a distributed computing framework where the distributed nodes have different communication capabilities, motivated by the heterogeneous networks in data centers and mobile edge computing systems. Following the structure of MapReduce, this framework consists of Map computation phase, Shuffle phase, and Reduce computation phase. The Shuffle phase allows distributed nodes to exchange intermediate values, in the presence of heterogeneous communication bottlenecks for different nodes (\emph{heterogeneous communication load constraints}). 
For this setting, we characterize the   \emph{minimum total computation load} and the  \emph{minimum worst-case computation load} in some cases, under the heterogeneous communication load constraints. 
While the total computation load depends on the sum of the computation loads of all the nodes,  the worst-case computation load depends on the computation load of a node with the heaviest job.
We show an interesting insight that, for some cases, there is a tradeoff between the minimum total computation load  and the minimum worst-case computation load, in the sense that both cannot be achieved at the same time. The achievability schemes are proposed with careful design on the file assignment and the data shuffling.  
Beyond the cut-set bound, a novel converse is proposed using the proof by contradiction.
For the general case, we identify two extreme regimes in which both the scheme with coding  and the scheme without coding are optimal, respectively.

\end{abstract}

\section{Introduction}

In recent years, with the availability of low-cost servers and big data, distributed computing systems have come to prominence within industrial sectors. Distributed computing frameworks such as  MapReduce\cite{DGmapreduce:2004}, Hadoop\cite{SKRC:10} and Spark\cite{ZCFSS:10} have been used in many applications that require complex computations, e.g., machine learning and distributed virtual reality (VR).

In distributed computing systems, since a complex computational task is split and assigned to distributed nodes (workers), communication is an important step that facilitates the information exchange across distributed nodes. 
In the MapReduce-based distributed computing framework  (cf.~\cite{DGmapreduce:2004}), data is first split and processed (called \emph{Map}) at the distributed nodes, and then the results are \emph{shuffled} (called \emph{Shuffle}), and  processed again (called \emph{Reduce}). As the amount of data and the number of nodes grow, heavy communication is required for data shuffling phase, which could lead to a non-negligible delay for the overall performance. 

Recently,  a significant number of works have focused on improving the performance of distributed computing system with the use of ``coding''\cite{LMA:15,  LMYA:17, LMAAllerton:16, LYMA:16, YLMA:17, EKF:17, AT:17, DCG:16, DCG:17,  Lee2015speeding, LSR:17, KSD:17, RPPA:19, RP:17, TLDK:17, AT:16,  SFZ:17,  ATallerton:16, PLSSM:18,PRPA:18,PLEisit:18,ZS:19, LTC:18, LCW:19,WCJ:19,YYW:18}.  Specifically, the works in \cite{LMA:15, LMYA:17, LMAAllerton:16, LYMA:16, YLMA:17} have introduced \textit{Coded Distributed Computing (CDC)} framework that utilizes the computational power of the distributed nodes with the addition of redundant jobs, which can reduce the communication load in the Shuffle phase. CDC has been developed for \emph{homogeneous} computing systems. However, distributed computing systems  are \emph{heterogeneous} in nature \cite{KWA:17}. In distributed computing systems, different nodes would have different computation capabilities, as well as different communication capabilities. 
For example, in data centers and mobile edge computing systems, the communication typically takes place over heterogeneous networks, in the presence of \emph{heterogeneous} communication bottlenecks that limit differently on different links of the networks \cite{PL:18,MYZHL:17,MB:17, CTJ:17}.

In this work, we consider a distributed computing framework with heterogeneous communication constraints, where different nodes have different communication capabilities. 
Based on the structure of MapReduce, this framework consists of Map computation phase,  Shuffle phase, and Reduce computation phase.
Specifically, the system seeks to compute $Q$ output functions for $N$ input files over $K$ distributed nodes, for some positive integers $Q, N$ and $K$. At first, the input files  are assigned to $K$ distributed nodes by design. The output functions are decomposed into some Map functions and Reduce functions.
  Each Map function then takes one file, e.g., file~$n$, as input and outputs  $Q$ intermediate values $\{a_{q,n}\}_{q=1}^{Q}$, that will be used  for the Reduce functions.
In the Shuffle phase, some intermediate values are shuffled among the distributed nodes. In this work we consider the scenarios  with \emph{heterogeneous} communication constraints that limit differently on the amounts of data to send (\emph{communication loads}) for different nodes,  captured by the parameters $L_1, L_2, \cdots, L_K$.
Given the heterogeneous communication load constraints, we seek to characterize the  \emph{minimum total computation load} (denoted by $\Mt$), as well as the \emph{minimum worst-case computation load} (denoted by $\Mwt$) of a distributed computing  system.  
In our setting, the \emph{computation load} of a node is defined by the number of files computed in the node, and $\Mt$ and $\Mwt$ are captured by the sum and the maximum one of all $K$ computation loads, respectively.
In a distributed computing system,  intuitively $\Mt$ is connected with the total resource consumption, while $\Mwt$ is connected with the overall latency because it is affected by the computation time of a node with the  heaviest job.

The main contribution of this work is  the \emph{information-theoretical} characterization  of the minimum  total computation load $\Mt$   and the minimum   worst-case computation load $\Mwt$ in some cases, for the distributed computing systems with heterogeneous communication load constraints. 
The results reveal that, in most of the cases,  $\Mt$  depends on the total communication load constraint parameter, which is defined as $L \defeq \sum_k L_k$, while $\Mwt$ depends on the individual communication load constraint parameters $L_1, L_2, \cdots, L_K$.
Table~\ref{tb:introtable} provides  a summary of $\Mt$ and  $\Mwt$ for three cases of parameters. 
 One can see that, given a fixed $L$,  $\Mt $ is fixed for all of these three cases but $\Mwt$ is not; $\Mwt$ depends on  individual $L_1, L_2, L_3$.

We also show an interesting insight that, in a certain region, there is a tradeoff between the \emph{minimum total computation load} $\Mt$   and the \emph{minimum worst-case computation load} $\Mwt$, in the sense that both cannot be achieved at the same time. For one example with  $(K=Q=3,  N=7, L_1=L_2=2, L_3 =14)$,  we have $\Mt = 7$ and $\Mwt = 4$ (see Table~\ref{tb:introtable}), but we prove that these two cannot be achieved at the same time. 

\begin{table}
\begin{center}
\caption{A summary of $\Mt$ and  $\Mwt$ for three cases of parameters, given $K=Q=3$.} \label{tb:introtable}
\begin{tabular}{cccccccc}
\toprule
Case & $N$  & $L_1$ & $L_2$ & $L_3$ & $L$   & $\Mt$ & $\Mwt$ \\ 
\midrule
$A$ & 7    & 2   &  2  &  14 & 18  &  7  &   4 \\    
$B$ & 7    & 2   &  4  &  12 & 18  &  7  &   3 \\      
$C$ & 7    & 6   &  6  &  6 & 18  &  7  &   3    \\  
\bottomrule
\end{tabular}
\end{center}
\end{table}

To prove our results, the achievability schemes are proposed with careful design on the file assignment and data shuffling, under the heterogeneous communication load constraints.   For the general case, we identify two extreme regimes in which both the scheme \emph{with coding}  and the scheme \emph{without coding} are optimal, respectively.  Note that for the scheme with coding, the transmitted symbol is usually a function (e.g., XOR function) of some information symbols.
In this work, a novel converse is proposed using the \emph{proof by contradiction}.
We show that in some cases, proof by contradiction is a very powerful approach to derive the optimal converse bound, which is strictly better than cut-set bounds and other converse bounds derived from the existing techniques.
Let us focus on one example with $(K=Q=3,  N=7, L_1=L_2=2, L_3 =14)$. For this example, the   \emph{minimum  worst-case computation load} is $\Mwt = 4$, which is achievable by a proposed scheme described in Fig.~\ref{egcontradict0}. 
For the converse,  proof by contradiction produces a novel bound $\Mwt \geq 4$, which is tighter than the cut-set bound ($\Mwt \geq 3$) and the bound derived from \cite{LMYA:17, KWA:17}. Note that the scheme depicted in Fig.~\ref{egcontradict0} is optimal even though coding is not used to broadcast the information symbols. When each node has sufficient communication capabilities, the nodes can use more communication load by sending out uncoded intermediate values which incur lesser computation load at the nodes. As coding increases the computation load, and  as we focus on  minimizing the computation load, coding of information symbols is not always beneficial. More details of this example can be found in Section~\ref{eg:worstcaseexample}.

The paper is organized as follows. Section \ref{sec:system} presents the system model. Section \ref{sec:mainresult} provides the main results of this work. 
Section \ref{sec:example} presents some examples of our proposed schemes. 
The proofs are provided in Sections~\ref{sec:achiK}-\ref{sec:contradp2}.  
Specifically, the achievability schemes are described in  Sections~\ref{sec:achiK}-\ref{sec:achiMwt3}, and the converse is provided in Section~\ref{sec:converse2}.  
The work is finally concluded in Section \ref{sec:conclusion}.  
Throughout this work, $|\bullet|$ denotes a cardinality of a set.  $(\bullet)^\T$  denotes the transpose operation. $[N_1: N_2]$ denotes a set of integers from $N_1$ to $N_2$, for some integers $N_2\geq N_1$.  If $N_2 <N_1$, then $[N_1: N_2] = \emptyset$. 
 $\F^{q}_{2}$ denotes a set of $q$-tuples of binary numbers.   $(\bullet)^+ = \max\{0, \bullet \}$. 
   $\mathbb{N}$ denotes the set of  natural numbers including $0$ and $\mathbb{N}^+$ denotes the set of positive natural numbers.
    $\uv[i]$ denotes the $i^{\text{th}}$ element of vector $\uv$.
$\ceil*{c}$ denotes the least integer that is greater than or equal to $c$. 
Similarly, $\floor*{c}$ denotes the largest integer that is smaller than or equal to $c$.   $[x \text{ mod } y]$ denotes a modulo operation that produces  the remainder after division of $x$ by $y$ for two positive numbers $x$ and $y$.
$\Hen(x)$ denotes the entropy of a random variable $x$.
$\oplus $ denotes a bitwise operation (XOR). 
If the XOR operation is over vectors (or matrices) then, the output is also a vector (or a matrix, respectively).

\begin{figure}
\centering
\includegraphics[scale=1.25, width=15cm]{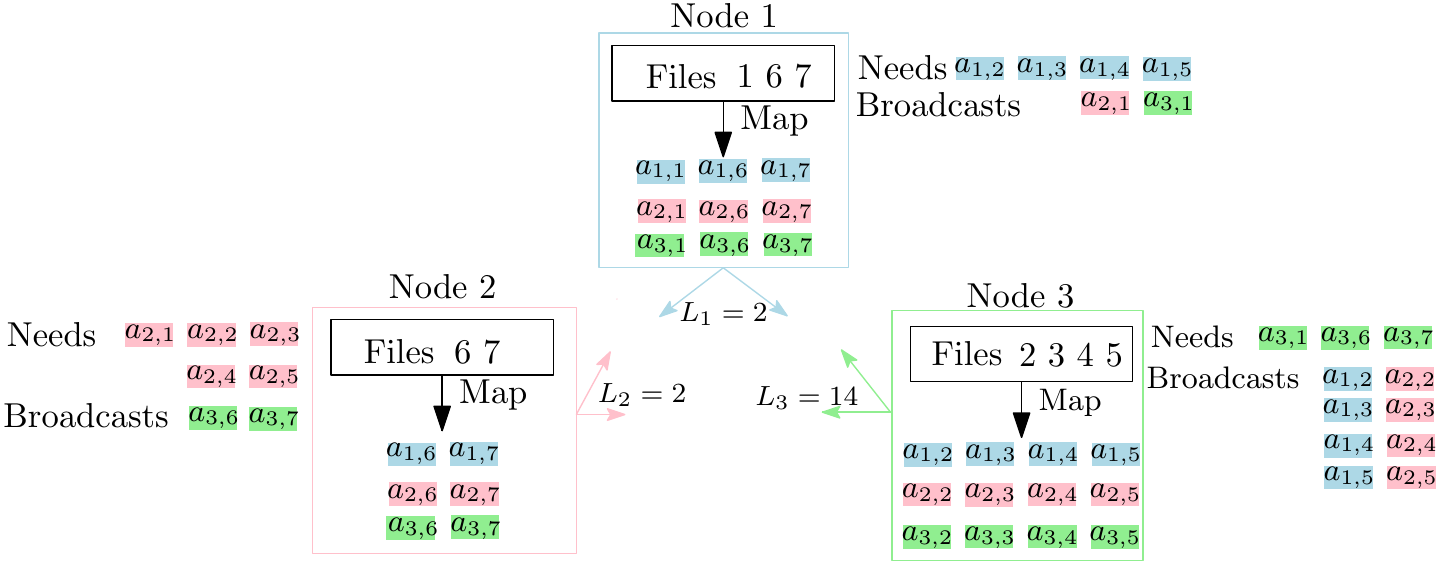}
\caption{A scheme for a distributed computing system with  $( K=Q=3,  N=7, L_1=L_2=2, L_3 =14)$. 
The scheme is optimal in terms  of the minimum \emph{worst-case computation load}, and the optimality proof is based on proof by contradiction. The intermediate values highlighted in blue, pink and green are the intermediate values required by Node~$1$, Node~$2$ and Node~$3$, respectively, to compute their output functions.
 } 
\label{egcontradict0}
\end{figure}

\section{System model \label{sec:system} }

We consider a distributed computing system based on a MapReduce framework (cf.~\cite{DGmapreduce:2004, LMYA:17}), in which  $K$ distributed nodes seek to compute $Q$ output functions using $N$ input files, for some $K,Q,N \in \mathbb{N}^+$, with $N\geq K$. The process of computing $Q$ output functions from $N$ input files can be broken down into three phases, that is, Map, Shuffle and Reduce (see Fig.~\ref{systemmodel}). Next, we will discuss these three phases in detail.

{\bf{\emph{Map phase}}}: In this phase,  $N$ input files, denoted by $\{w_n \in \mathbb{F}^{F}_2: n \in [1:N]\}$, are first assigned to $K$ nodes, for some  $F \in  \mathbb{N}^+$.  Let $\Mc_k \subseteq [1:N]$ denote a set of indices of the files assigned at Node~$k$, $k\in [1:K]$. Let \[M_{k}\defeq|\mathcal{M}_k|.\] In our setting,   $M_{k}$  denotes the \emph{computation load} of Node~$k$. 
For each input file $w_n,  n \in \Mc_k$, Node~$k$ generates $Q$ \emph{intermediate values}, denoted by $\{a_{q,n} \in \mathbb{F}_2^{B}: q \in [1:Q]\}$, for some $B \in \mathbb{N}^+$, where
\[a_{q,n}=  g_{q, n}(w_n).\]
In this setting, $g_{q, n}(w_n)$ is a \emph{Map function} that maps the input file $w_n$ to a length-$B$ value $a_{q,n}$.
We assume that all the intermediate values are  independent and identically distributed (i.i.d.) random variables uniformly distributed over $\mathbb{F}^{B}_2$. 
The realization of $a_{q,n}$ is determined by the input file $w_n$ and the Map function $g_{q,n}$, for $q \in [1:Q]$ and $n \in [1:N]$.

{\bf{\emph{Shuffle phase}}}: 
In our setting, each node is responsible for computing a subset of output functions.  
We use $\mathcal{W}_k$ to denote a set of indices of the output functions computed at Node~$k$, $k \in [1:K]$. In order to complete the whole computation,  $K$ distributed nodes need to exchange intermediate values in the Shuffle phase. 
Specifically, in this phase Node~$k$ multicasts to  the other nodes a message  
\begin{align}
x_{k} = f_{k}(a_{\mathcal{W}_k^c,\Mc_k}),  \label{eq:xkdef}
\end{align}
which is a deterministic function of the  intermediate values cached at Node~$k$ and intended for the other nodes, where $a_{\mathcal{W}_k^c,\Mc_k} \defeq \{a_{q, n }:  q \in [1:Q] \setminus  \mathcal{W}_k, n \in \Mc_k \}$,  for $k \in [1:K]$.
At the end of this phase, Node~$k$ has all the intermediate values needed for its Reduce functions in the next phase.

{\bf{\emph{Reduce phase}}}: In this phase, each node proceeds to compute a set of final output values by using the intermediate values  acquired from Shuffle phase and the local intermediate values computed from Map phase. 
Specifically, for each $q\in \mathcal{W}_k$, Node~$k$ computes the final output value 
\begin{align}
 b_{q}=  \varphi_{q}(a_{q,1},a_{q,2}, \cdots,  a_{q, N}),  \non 
\end{align}
where  $ \varphi_q$ is a \emph{Reduce} function that maps the intermediate values  $\{a_{q,1}, a_{q,2}, \cdots,  a_{q, N}\}$  into an  output value  $b_{q} \in \mathbb{F}_2^{B'}$,   for some  $B' \in \mathbb{N}^+$, $k \in [1:K]$.
We assume  a symmetric job assignment, i.e., each node calculates $Q/K$ Reduce functions, for $Q/K \in \mathbb{N}^+$,  that is, 
\begin{align}
|\mathcal{W}_1|=|\mathcal{W}_2|=\cdots=|\mathcal{W}_K| = Q/K,   \label{eq:funcsym}
\end{align}
and $\mathcal{W}_k \cap \mathcal{W}_j= \varnothing$ for any $k, j \in [1:K]$, $k\not =j$.

We consider a \emph{communication load constraint} for Node~$k$ such that
\begin{align}
& \Hen(x_{k}) \leq    L_{k}\cdot QB/K   ,  \quad k=1,2,\cdots, K,  \label{eq:comc} \\
\text{or equivalently, } & K \cdot \Hen(x_k)/QB \le L_k,  \quad k=1,2,\cdots, K,   \label{eq:comc2}
\end{align}
for some $L_{k} \in  \mathbb{N}$, where $\Hen(x_{k})$ denotes the entropy of the message $x_{k}$ that will be multicast from Node~$k$ to the other nodes. The constraint in \eqref{eq:comc} can be considered as the communication bottleneck of Node~$k$, i.e., the maximum number of  bits of the information that can be sent. The  level of this bottleneck for Node~$k$ is reflected by the parameter $L_{k}$.
We also let
\begin{align} \label{eq:sumL}
L \defeq \sum_{k=1}^{K} L_k
\end{align}
be a parameter of  the \emph{total communication load constraint} of all $K$ nodes. 

In our setting, any MapReduce scheme consisting of Map, Shuffle and Reduce phases should be designed under the communication load constraint in \eqref{eq:comc}.
For any MapReduce scheme,  the total computation load  and  the worst-case computation load are defined by 
\begin{align}
\Mtn  =    \sum_{k=1}^{K} |\mathcal{M}_k|  \quad 
\text{and} \quad \Mwtn  =   \max \Big\{\! |\mathcal{M}_{1}|,\cdots,|\mathcal{M}_{K}|\Big\},
\end{align}
respectively.
In this work, we consider the \emph{minimum} (optimal) total computation load  and  \emph{minimum} (optimal) worst-case computation load of the system, which are defined as follows.

\begin{figure}
\centering
\includegraphics[width=10cm]{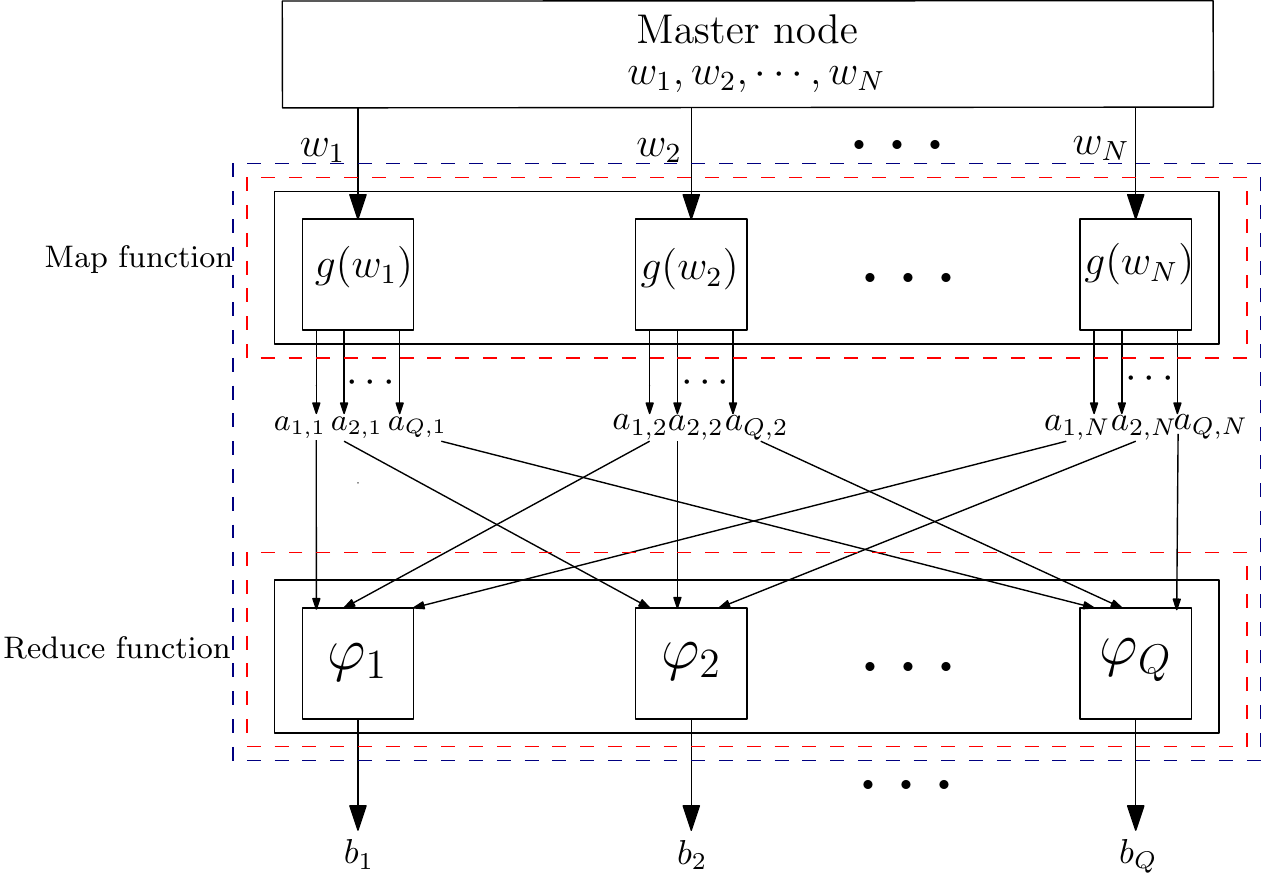}
\caption{A distributed computing system model.  The system consists of Map phase, Shuffle phase and Reduce phase. The Shuffle phase allows distributed nodes to exchange intermediate values, in the presence of heterogeneous communication load constraints. 
 } 
\label{systemmodel}
\end{figure}

\begin{definition}  \label{def:totalcomputation}
Given the communication load constraints with 
${L}_: \defeq \{{L}_1, {L}_2, \cdots, {L}_K\}$, 
and  the  number of files $N$,  the minimum \emph{total computation load} is defined by 
\begin{align*}
\Mt(N,{L}_{:})  = \min_{\Mc_{:} \ : \ \mathbb{S}(N, {L}_{:}, \Mc_{:}) \subseteq \mathbb{S}(N, {L}_{:})}   \sum_{k=1}^{K} |\mathcal{M}_k|,
\end{align*}
where $\Mc_{:}\defeq \{\Mc_{1}, \Mc_{2}, \cdots, \Mc_{K}\}$; $\mathbb{S}(N, {L}_{:})$ denotes the set of all MapReduce schemes with parameters  $(N, {L}_{:})$; and $\mathbb{S}(N, {L}_{:}, \Mc_{:})$ denotes a class of MapReduce schemes  whose file assignment is determined by $\Mc_{:}$.
\end{definition}
\vspace{3pt}

\begin{definition}  \label{def:computation}
Given the communication load constraints with 
${L}_: \defeq \{{L}_1, {L}_2, \cdots, {L}_K\}$, 
and  the number of files $N$,  the  minimum \emph{worst-case computation load} is defined by
\begin{align*}
&\Mwt(N,{L}_{:}) =  \!\!\min_{\Mc_{:} \ :  \ \mathbb{S}(N, {L}_{:}, \Mc_{:}) \subseteq \mathbb{S}(N, {L}_{:})}  \!\! \!\! \! \! \max \! \Big\{\! |\mathcal{M}_{1}|,\cdots,|\mathcal{M}_{K}|\Big\}.
\end{align*}
\end{definition}

\vspace{5pt}

In a distributed computing system,  intuitively $\Mt$ is connected with the total resource consumption, while $\Mwt$ is connected with the overall latency because it is affected by the computation time of a node with the  heaviest job.

For notational convenience, we define $\Sc_\Ac$ as the indices of the files placed in each of the nodes indexed by $\Ac$ but not in the other nodes indexed by $\Ac^c$, that is,
\begin{align} \label{eq:Sdefine}
\Sc_\Ac \defeq \cap_{i \in \Ac} \Mc_i \setminus \cup_{j \in \Ac^c }\Mc_j
\end{align}
for a set $\Ac \subseteq [1:K]$ and $\Ac^c \defeq [1:K] \setminus \Ac$. 
For example, $\Sc_{\{1,2\}} = \Mc_1 \cap \Mc_2 \setminus \Mc_3 \cup \Mc_4$ for the setting of $K=4$. 
The cardinality of $\Sc_\Ac$ is denoted by $S_\Ac \defeq |\Sc_\Ac|$.
For simplicity we will use notation $S_{12}$ to represent $S_{\{1,2\}}$ and similar notations are used for the other set $\Ac$.

\section{Main results  \label{sec:mainresult}}

In this section, we provide the main results of this work. 
We will begin with the two-node ($K=2$) and three-node ($K=3$)  distributed computing systems defined in Section~\ref{sec:system}, and  provide an  information-theoretical characterization of the minimum (optimal) total computation load, as well as the minimum (optimal) worst-case computation load. 
After that, we will focus on the distributed computing system with a general $K$.

\subsection{The case with $K=2$}

For the two-node  distributed computing system, the results are provided in the following  theorem.

\vspace{3pt}
\begin{theorem} [$K=2$] \label{thm:K2}
For a  two-node distributed computing system defined in Section~\ref{sec:system}, the minimum total computation load and the minimum worst-case computation load are characterized by
\begin{align}
\Mt  &=  \max \{N, \ 2N- L\} ,  \\ 
\Mwt &= N - \min \big\{{L}_1, {L}_2,  N - \ceil*{N/2}\big\}.   
\end{align}
\end{theorem}

The results of Theorem~\ref{thm:K2}  are achieved by the same scheme described in Section~\ref{sec:achiKsame}, which, not only achieves the  minimum total computation load, but also achieves the minimum worst-case computation load for the setting with $K=2$. 
The converse proof of Theorem~\ref{thm:K2} is provided in Section~\ref{sec:ThmTK2}. 
Theorem~\ref{thm:K2} reveals that, in this setting with $K=2$, $\Mt$ depends on the total communication load constraint, while $\Mwt$ depends on the individual communication load constraints.

\subsection{The case with $K=3$}

We proceed to extend the above results to the setting with three nodes ($K=3$). 
Note that, when the number of nodes is increased, the problem in our setting becomes  more challenging. This is because the optimal solution  to our problem (e.g.,  the optimal achievability scheme) must satisfy  $K$ heterogeneous communication load constraints.
For the setting with three nodes, the result on the minimum total computation load is given in the following theorem.

\begin{theorem} [Total, $K=3$] \label{thm:K3} 
For a  three-node distributed computing system, and given  ${L}_k/2 \in \mathbb{N},  \forall k \in \{1,2,3\}$, the minimum total computation load is characterized by
\begin{align} \label{eq:3bound}
\Mt= \max \Big\{N, \ \ceil[\Big]{\frac{7N}{3} - \frac{2L}{3}},  \ 3N-2{L}\Big \}.
\end{align}
\end{theorem}
\vspace{3pt}

The achievability and the converse proofs of this theorem are presented in Section~\ref{sec:achiK} and Section~\ref{sec:ThmTK3}, respectively.
Next, we will focus on the minimum worst-case computation load of a three-node distributed computing system. We will provide some converse bounds at first and then discuss some optimal cases.

\begin{lemma} [Worst, $K=3$]  \label{lm:converseMwt}
For a  three-node distributed computing system defined in Section~\ref{sec:system},  the minimum worst-case computation load  is lower bounded by
\begin{align*}
\Mwt \geq  \max \Bigg\{& \ceil*{\frac{N}{3}},   \   N - \min_{i\neq j}\{{L}_i + {L}_j\},     \ceil*{\frac{N}{2}-\frac{\min_k\{L_k\}}{4}}, \ceil*{\frac{\ceil*{\frac{7N-2L}{3}}}{3}}  \Bigg\}.        
\end{align*}
\end{lemma}

The proof is provided in Section~\ref{sec:lemWK3}. 
In this proof we use the ``cut-set'' technique and the other existing technique.
The following lemma provides a novel bound that is derived from the proof by contradiction. 

\begin{lemma} [Proof by Contradiction] \label{prop:convbound5}  
For a  three-node  distributed computing system, 
the minimum worst-case computation load  is lower bounded by 
\[\Mwt \geq \beta^*,\] 
where $\beta^* $ is defined by the following optimization problem
\begin{align} 
\beta^* \!=\! \max  \  & \beta, \non \\
  \text{s.t.}   \   & \beta \leq \ceil[\bigg]{\ceil[\Big]{ \bigl(7N-2 \cdot \sum_{k=1}^3   \min\{L_k, \  2(\beta - 1)\} \bigr)\Big/  3}\Big/3},  \non \\
   & \beta \in \mathbb{N}.  \non
\end{align}
\end{lemma}
\vspace{3pt}

The bound in Lemma~\ref{prop:convbound5} is proved in Section~\ref{sec:lemWK3Con}.
In some cases, this bound is strictly better than all the bounds in Lemma~\ref{lm:converseMwt}. 
Let us focus on one example defined by the parameters $(K=Q=3,  N=7, L_1=L_2=2, L_3 =14)$. For this example,  proof by contradiction produces a novel bound $\Mwt \geq 4$ (Lemma~\ref{prop:convbound5}), which is strictly tighter than  all the bounds  in Lemma~\ref{lm:converseMwt} ($\Mwt \geq 3$).

Let us now  provide some  cases in which we have the optimal characterization of the worst-case computation load. 
To discuss the optimality of the converse bounds, we define three conditions as follows:
\begin{align} 
\text{Condition~1:} &\quad  \min_{k} \{{L}_k\} \geq 2\ceil*{\frac{N}{3}},   \label{eq:cond1} \\
\text{Condition~2:} & \quad   {L} \leq \frac{N}{2},    \label{eq:cond2} \\
\text{Condition~3:} & \quad   2 \leq \min_k\{{L}_k\} \leq \frac{2N}{3}  \  \text{and}  \ \frac{3N -  \min_{i\neq j}\{{L}_i + {L}_j\} }{5}  \leq \!  \ceil*{\frac{N}{2} - \frac{\min_k\{L_k\}}{4}} \! \leq \frac{\max_k\{{L}_k\}}{2}.     \label{eq:cond3}
\end{align}

The results are shown in the following Proposition~\ref{prop:K3optimal}. 

\begin{proposition} [Worst, $K=3$] \label{prop:K3optimal}
For a  three-node distributed computing system, under each condition in \eqref{eq:cond1}, \eqref{eq:cond2} and \eqref{eq:cond3},  the minimum worst-case computation load is respectively characterized by
\begin{align} 
\text{for Condition~1:} \quad \quad \Mwt &= \ceil*{\frac{N}{3}},    \label{eq:K3worst1} \\ 
\text{for Condition~2:} \quad \quad   \Mwt &= N - \min\{{L}_1 + {L}_2, {L}_2+{L}_3, {L}_1 + {L}_3\}, \label{eq:K3worst2}\\
\text{for Condition~3:} \quad \quad   \Mwt &= \ceil*{\frac{N}{2} - \frac{\min_k\{L_k\}}{4}}.   \label{eq:K3worst3} 
\end{align}

\end{proposition}

The achievability of Proposition~\ref{prop:K3optimal} is shown in Section~\ref{sec:achiMwt3} and the  converse is directly  from Lemma~\ref{lm:converseMwt}.
Fig.~\ref{fig:M1} depicts the performance of  $\Mwt$ vs. $L_1$ for the  case with $(K=Q=3, N=14, L_1/2=L_2/2=L_3/4 \in  \mathbb{N}^+ )$. 
In this case there is a tradeoff between the minimum worst-case computation load $\Mwt$  and the total communication load constraint $L$ ($L=4L_1$ in this case).
The characterization of $\Mwt$  in Fig.~\ref{fig:M1} stems from the above lemmas and proposition. Specifically, for the point of $( L_1=L_2=L_3/2 = 8; \Mwt = 6)$, the converse follows from  the result of Lemma~\ref{prop:convbound5} that is derived from the proof by contradiction.

\begin{figure}
\begin{center}
\begin{tikzpicture}[scale = 1]
    \begin{axis}[
        xlabel={$L_1$}, ylabel={$\Mwt$}, 
        xmin=0, xmax=22,
	    ymin=3, ymax=12,
    	xtick={0,2,4,6,8,10,12,14,16,18,20,22},
	    ytick={3,4,5,6,7,8,9,10,11,12},
    	legend pos=north east,
    ]
        \addplot table{\mytable};
        \legend{$N=14$ and $L_1=L_2=L_3/2$ }
   \end{axis}
\end{tikzpicture}
\vspace{-8pt}
  \caption{$\Mwt$ vs. $L_1$ for the case with $(K=Q=3, N=14, \frac{L_1}{2}=\frac{L_2}{2}=\frac{L_3}{4} \in  \mathbb{N}^+ )$. } \label{fig:M1}
\end{center}
\end{figure}
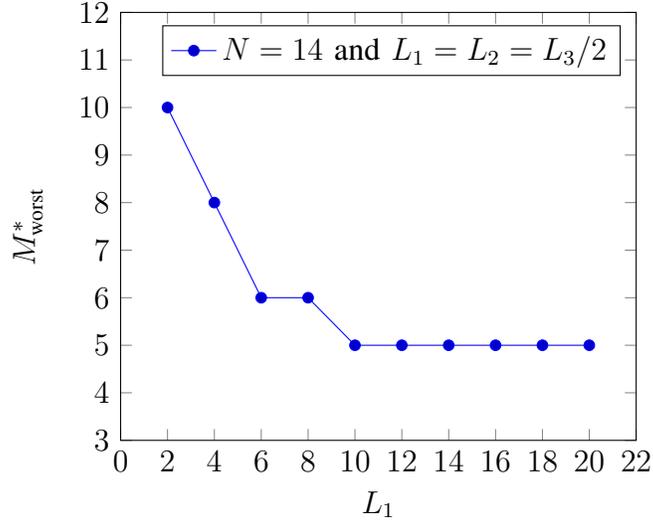

\subsection{The general case with $K \geq 3$}

Let us now provide the results of the distributed computing system,  for the case with $K\geq 3$. 
For this general case, we identify two extreme regimes in which both the scheme with coding  and the scheme without coding are optimal, respectively.

\begin{proposition}[coding, general $K$] \label{prop:KsmallL}
For a $K$-node distributed computing system, if the following condition is satisfied
 \begin{align} 
{L}\leq \frac{N}{K-1} ,    \label{eq:condG1}
\end{align}
 then the minimum total computation load and the minimum worst-case computation load are characterized, respectively, as
 \begin{align} 
\Mt &= KN - (K-1){L},  \\
\Mwt &= N - {L} + \max_{k} \{{L}_k\} .   
\end{align}
\end{proposition}
The results of Proposition~\ref{prop:KsmallL}  are achieved by a scheme \emph{with coding} (bitwise XOR operation of intermediate values) described in Section~\ref{subsec:coding}.  
The converse is described in Section~\ref{sec:converse2}.

\begin{proposition}[no coding, general $K$] \label{prop:KlargeL}
For a distributed computing system consisting of $K$ nodes, if the following condition is satisfied
 \begin{align} 
\min_k \{{L}_k\} \geq (K-1) \cdot \ceil*{\frac{N}{K}}, \label{eq:condG2} 
\end{align}
 then the minimum total computation load and the minimum worst-case computation load are characterized, respectively, as
 \begin{align}
\Mt &= N,  \\ 
\Mwt &= \ceil*{\frac{N}{K}}. 
\end{align}
\end{proposition}

The results of Proposition~\ref{prop:KlargeL}  are achieved by a scheme \emph{without coding} described in Section~\ref{subsec:nocoding}. 
The converse is described in Section~\ref{sec:converse2}.

\subsection{Tradeoff between $\Mt$ and $\Mwt$}

In a distributed computing system, due to the relation with total resource consumption and overall latency, one might want to minimize  both the total  computation load  and the worst-case computation load  as much as possible.  At this point, one interesting question is raised: can we achieve  $\Mt$ and $\Mwt$  at the same time? 

For  a $K$-node distributed computing system,  we propose a scheme that indeed can  achieve   $\Mt$ and $\Mwt$  at the same time (see Section~\ref{sec:achiKsame}).  Given the condition in \eqref{eq:condG1} or the condition in \eqref{eq:condG2} for some $K$,   $\Mt$ and $\Mwt$  can be achieved at the same time (see Propositions~\ref{prop:KsmallL} and  \ref{prop:KlargeL}).  
However, for some cases it is not always true --- the following theorem reveals an instance in which  $\Mt$ and $\Mwt$  cannot be achieved at the same time.

\begin{theorem}[$\Mt$ vs. $\Mwt$]\label{thm:tradeoff}
For the case with a three-node distributed computing system, if the following condition is satisfied
\begin{align} \label{eq:tradeoffeq}
\Mt & <  3N- 2\Mwt - \min_{i\neq j}\{L_i + L_j\}, 
\end{align}
then  $\Mt$ and $\Mwt$  cannot be achieved at the same time.
\end{theorem}

We prove in Section~\ref{sec:contradp2} that, when the condition in \eqref{eq:tradeoffeq} is satisfied,  $\Mt$ and $\Mwt$  cannot be achieved at the same time. 
For one example with  $(K=Q=3,  N=7, L_1=L_2=2, L_3 =14)$,  we have $\Mt = 7$ and $\Mwt = 4$, but we prove that these two cannot be achieved at the same time.

\section{Examples \label{sec:example}}

This section provides examples on the proposed achievability schemes for the minimum total computation load and the minimum worst-case computation load. The first example consists of two schemes where the first scheme is designed for the minimum total computation load and the second scheme is designed for the minimum worst-case computation load, based on same parameters ($L_k, N, Q, K$). For this example, we conclude later that both the minimum total computation load and minimum worst-case computation load cannot be achieved using the same scheme. In the other two examples, we look into the achievable schemes for the minimum total computation load under symmetric communication load constraints and under asymmetric communication load constraints, respectively, and we conclude an interesting insight later. In all the examples (with $Q=K$) described in this section, we consider the output function~$k$ to be computed at Node~$k$, for $k \in [1:K]$ without loss of generality (WLOG). 

\subsection{Example with ($L_1=L_2=2$, $L_3=14$, $N=7$, $Q=K=3$)} \label{sec:exmp31}
Let us consider the example with ($L_1=L_2=2$, $L_3=14$, $N=7$, $Q=K=3$) to explain the two different schemes to achieve the minimum total computation load and the minimum worst-case computation load, respectively. First, we look into the scheme for the minimum total computation load and then, investigate the scheme for the minimum worst-case computation load.

\vspace{5pt}
\subsubsection{Scheme for the minimum total computation load} \label{subsec:totalscheme1}

The scheme design can be explained using different phases of MapReduce framework. The details of each phase are described below.
\begin{figure}
\centering
\includegraphics[scale=1.25, width=15cm]{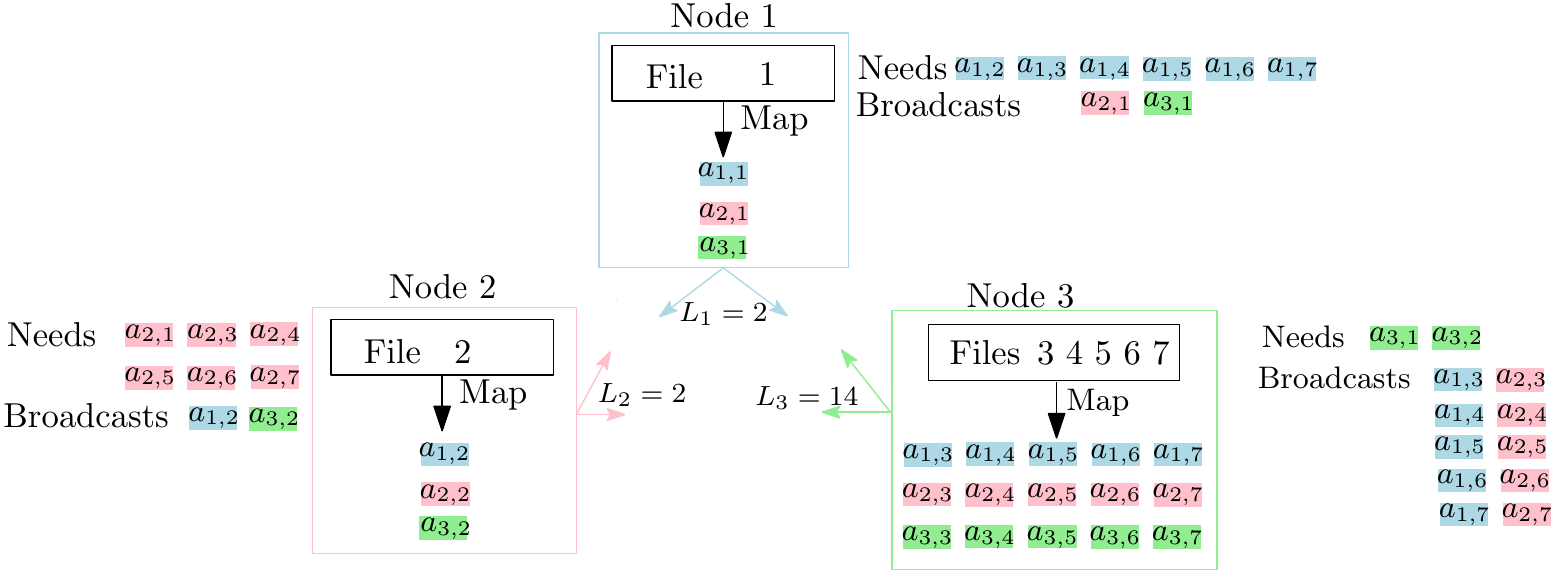}
\caption{A scheme for a distributed computing system with  $( K=Q=3,  N=7, L_1=L_2=2, L_3 =14)$. 
The scheme is optimal in terms  of the minimum \emph{total computation load}. 
 } 
\label{fig:example_total}
\end{figure}

{\emph{File placement}}: In our design, file placement is started from the node with the lowest communication load (Node~$1$), continued to Node~$2$ and then to Node~$3$. Note that our focus here is to minimize the computation load that is possible by compensating the communication load which, however is limited for each node. Files are placed exclusively at Node~$k$ if it has enough communication load ($L_k$) to send the intermediate values of the associated files required by the other nodes. Note that only one file is placed at one time and for the next file, next node with enough communication load is chosen. Here, file $w_1$ is placed at Node~$1$, then file  $w_2$ is placed at Node~$2$ and then file $w_3$ is placed at Node~$3$. As Nodes~$1$ and $2$ already have their communication load occupied, the remaining files, $\{w_n\}_{n=4}^7$, are placed at Node~$3$. We design the file placement in this way to avoid placing all the files in one node, i.e., approaching to a centralized system. With this approach, the file placement can be represented as:  $\Mc_{1}=\{1\}$,  $\Mc_{2}=\{2\}$ and $\Mc_{3}=\{ 3, 4, 5,6,7\}$,  with $M_1=1$, $M_2=1$ and $M_3=5$. 

{\emph{Map phase}}: Based on the above file placement, Node~$k$ computes the intermediate values  
$\{a_{q,n}: q\in [1: Q], n \in  \mathcal{M}_k \}$, for  $k=1,2, 3$. Node~$k$ still needs intermediate values $\{a_{q,n}: q\in \Wc_k, n \in  [1:N] \setminus \mathcal{M}_k \}$ for $k = 1,2,3$ to compute the output function $\Wc_k$.

{\emph{Shuffle phase}}: To fulfill the required intermediate values for the output function, $\Wc_k$, each node $k$ sends out intermediate values to the other nodes utilizing the limited communication load constraints. Here, in this example, Node~$1$ sends intermediate values $\{a_{2,1}, a_{3,1}\}$ to Nodes~$2$ and $3$. Since,  $\{a_{2,1}, a_{3,1}\}$ carries at most $2B$ bits of information, the communication load of Node~$1$ satisfies the constraint in \eqref{eq:comc2} for this case with $L_{1} = 2$, that is,  \[ \frac{K \cdot \Hen(x_{1}(a_{2,1},a_{3,1})) }{QB} \leq \frac{K \cdot \Hen(a_{2,1}, a_{3,1})}{QB} \leq  \frac{2K  B }{QB} = L_1=  2. \]   The first inequality follows from identity of $\Hen(f(e)) \le \Hen(e)$ for a deterministic function $f(e)$. Similarly, Node~$2$ sends out the message of $\{a_{1,2}, a_{3,2}\}$ of at most $2B$ bits to the other nodes, which satisfies the communication load constraint of Node~$2$ with $L_2=2$ while Node~$3$ sends out the message of $\{a_{1, n},  a_{2, n }\}_{n=3}^7$ of at most $10B$ bits to the other nodes which is less than the communication load constraint of Node~$3$, $L_3=14$.

{\emph{Reduce phase}}: 
Finally, in this phase, Node~$k$, $k\in [1:K]$ collects all the intermediate values $ \{a_{q,n}: q \in \mathcal{W}_k, n \in[1: N]\}$, from the Map phase and the Shuffle phase, as the inputs to compute the Reduce function.
For this proposed scheme, the total computation load is
\begin{align*}
\Mtn= M_1 + M_2 + M_3 = 7, 
\end{align*}
which is optimal. This scheme is an example of the general scheme which achieves the minimum total computation load described in Section~\ref{sec:achiK}.

\vspace{5pt}
\subsubsection{Scheme for the minimum worst-case computation load} \label{eg:worstcaseexample}
It turns out that the above achievable scheme optimal for the total computation load is, however, not optimal for the worst-case computation load. For the above scheme, the worst-case computation load is given by \[\Mwtn = \max\{M_1, M_2, M_3\} = 5.\]
With the given communication load constraints ($L_1 = L_2 = 2, L_3 = 14$), Node~$1$ and Node~$2$ can each send out at most $2$ intermediate values while Node~$3$ is capable of sending out $14$ intermediate values. Intuitively, as the communication capability of Node~$3$ is quite higher than that of other nodes, placing most of the files at Node~$3$ seems a straightforward option. However, when we focus on the worst-case computation load, and in the presence of limited communication load constraints, the scheme must be designed carefully to minimize the worst-case computation load. The new scheme designed for the minimum worst-case computation load for the given parameters is illustrated in Fig.~\ref{egcontradict0} and the details of the scheme are described as follows.

{\emph{File placement}}:  
Since Node~$3$ is  allowed to have higher communication load ($L_3 = 14$) compared to the other nodes $(L_1=L_2=2)$ and as we focus on the worst-case computation load, the scheme is designed to allow more files to be placed at the node that incur  the highest communication load, i.e., Node~$3$ in this case. 
 In this scheme, we design the file placement such that  $\Mc_{1}=\{1,6,7\}$,  $\Mc_{2}=\{6,7\}$ and $\Mc_{3}=\{2, 3, 4, 5\}$,  with $M_1=3$, $M_2=2$ and $M_3=4$. Instead of utilizing all the communication load in Node~$3$, redundant files have been placed at Node~$1$ and Node~$2$ to reduce the worst-case computation load. Comparing this scheme to the one for the total computation load (see Section \eqref{subsec:totalscheme1}), one extra computation load has been removed from Node~$3$. However, in this case, the total computation load ($\Mtn=9$) is increased as we focus on minimizing the worst-case computation load.

{\emph{Map phase}}: In this phase, Node~$k$ generates the following  intermediate values  
$\{a_{q,n}: q\in [1: Q], n \in  \mathcal{M}_k \}$, for  $k=1,2, 3$.

{\emph{Shuffle phase}}: 
In the Shuffle phase, Node~$1$ broadcasts the message of $\{a_{2,1}, a_{3,1}\}$ to the other nodes. Since  $\{a_{2,1}, a_{3,1}\}$ carries at most $2B$ bits of information, the communication load of Node~$1$ satisfies the constraint in \eqref{eq:comc2} for this case with $L_{1} = 2$, that is,  \[ \frac{K \cdot \Hen(x_{1}(a_{2,1},a_{3,1})) }{QB} \leq  \frac{2K  B }{QB} = L_1=  2. \]  
Similarly, Node~$2$ broadcasts the message of $\{a_{3,6}, a_{3,7}\}$.  Since it carries at most $2B$ bits of information, the communication load of Node~$2$ satisfies the constraint in \eqref{eq:comc2} for this case with $L_2 = 2$. 
Node~$3$ broadcasts the message of $\{a_{1, n},  a_{2, n }\}_{n=2}^5$.  Since this message carries at most $8B$ bits of information, the communication load of Node~$3$ satisfies the constraint in \eqref{eq:comc2} for this case with $L_3 = 14$. 

{\emph{Reduce phase}}: 
In this phase, Node~$k$, $k\in [1:K]$, has all the intermediate values $ \{a_{q,n}: q \in \mathcal{W}_k, n \in[1: N]\}$ as inputs to compute its Reduce function. 
For this proposed scheme, the worst-case computation load is
\begin{align*}
\Mwtn=  \max \{M_1, M_2, M_3\} = 4, 
\end{align*}
which turns out to be optimal.  Note that for this setting, the converse proof is based on proof by contradiction (see Lemma~\ref{prop:convbound5}).

Compared to the scheme in \ref{subsec:totalscheme1}, for the scheme in \ref{eg:worstcaseexample}, the computation load of Node~$3$, the node with the heaviest job,  is lower which reduces the worst-case computation load.
From the above examples, we can see that either scheme cannot achieve both the minimum total computation load and the minimum worst-case computation load at the same time. Also note that the total computation load depends upon the total communication load constraint parameter, $L$ and the worst-case computation load depends upon individual communication load constraint parameters, $L_k$,  $k \in \mathbb{N}^+$, which is also illustrated in Table~\ref{tb:introtable}.

One intriguing aspect of the above examples is that the scheme designs don't include coded intermediate values. This is because the total communication load, $L$, is sufficient to employ the uncoded intermediate values for the given number of files, $N$. The above scheme designs are the examples of the achievability schemes which are explained later in Sections~\ref{sec:achiK} and \ref{sec:achiMwt3}. 

\begin{figure}
\centering
\includegraphics[scale=1.25, width=15cm]{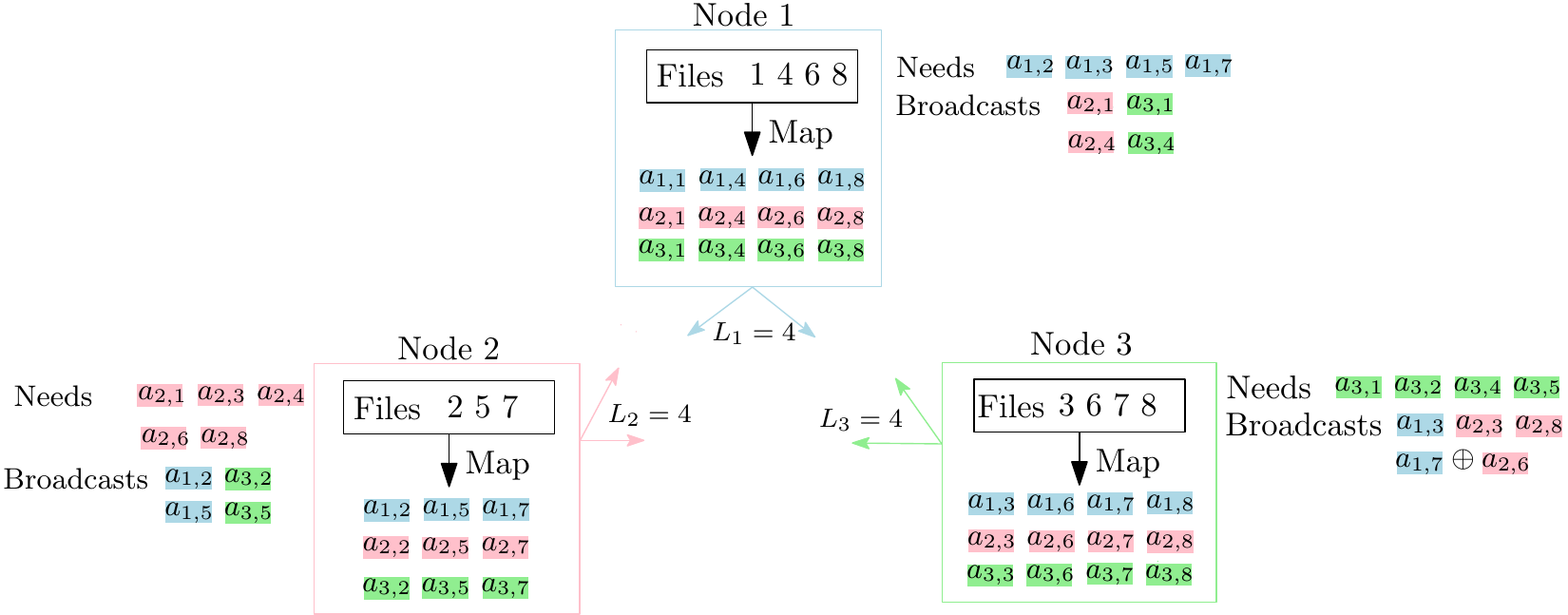}
\caption{A scheme for a distributed computing system with  symmetric communication load constraints $( K=Q=3,  N=8, L_1=L_2=L_3 =4)$. 
The scheme is optimal in terms  of the minimum \emph{total computation load}.
 } 
\label{fig:example_symmetric}
\end{figure}

\subsection{Example with ($L_1=L_2=L_3=4$, $N=8$, $Q=K=3$)} \label{sec:symmetric}
Here, in this example, we look at the minimum total computation load for the homogeneous communication load constraints ($L_1 = L_2 = L_3$). The scheme design for the example with ($L_1=L_2=L_3=4$, $N=8$, $Q=K=3$) is shown in Fig.~\ref{fig:example_symmetric}. Similar to the previous example, the scheme details are described for each phase.

{\emph{File placement}}:  Since all the nodes have the same communication load constraint ($L_1=L_2=L_3=4$), the file placement can start from any node. Here, we choose Node~$1$ at first in our design. We begin the file placement starting from $w_1$ and place it at Node~$1$. Then, we move onto the next nodes, Node~$2$ and then Node~$3$. The placement is done in a circular manner such that after the file placement at Node~$3$, next node chosen is Node~$1$. Files $\{w_n\}_{n=1}^5$ are placed in a circular manner in all the nodes such that Node~$1$ gets files $w_1$ and $w_4$, Node~$2$ gets files $w_2$ and $w_5$ and Node~$3$ gets file $w_5$. To satisfy the communication constraints at all the nodes, files $\{w_n\}_{n=6}^8$ are redundantly placed at two of the three nodes. For this example, the file placement is designed as follows:
\begin{align*}
&\Mc_1 = \{1,4,6,8\} \text{ with }  M_1 =4,\\
&\Mc_2 = \{2,5,7\} \text{ with }   M_2= 3,\\
&\Mc_3 = \{3,6,7,8\}  \text{ with } M_3 = 4.
\end{align*}
{\emph{Map phase}}: Node~$k$ generates the following  intermediate values  
$\{a_{q,n}: q\in [1: Q], n \in  \mathcal{M}_k \}$, for  $k=1,2, 3$. Node~$k$ still requires intermediate values $\{a_{q,n}: q\in \Wc_k, n \in  [1:N] \setminus \mathcal{M}_k \}$ for $k = 1,2,3$ to compute the output function $\Wc_k$.

{\emph{Shuffle phase}}: With the communication load constraints, $L_1=L_2=L_3=4$, each node can deliver at most $4B$ bits of information satisfying the communication constraint in \eqref{eq:comc2}. Specifically, Node~$1$ sends  $\{a_{2,1}, a_{3,1}, a_{2,4}, a_{3,4}\}$, Node~$2$ sends $\{a_{1,2}, a_{3,2}, a_{1,5}, a_{3,5}\}$ and Node~$3$ sends $\{a_{1,3}, a_{2,3}, a_{2,8}, a_{1,7} \oplus a_{2,6}\}$. The symbols sent by each node carry at most $4B$ bits of information. With the coded intermediate value $a_{1,7} \oplus a_{2,6}$, Nodes~$1$ and $2$ can decode the required information symbols, $a_{1,7}$ for Node~$1$ and $a_{2,6}$ for Node~$2$ by using the side information, $a_{2,6}$ from Node~$1$ and $a_{1,7}$ from Node~$2$, generated in the Map phase, respectively. The careful design of the file placement, Map and Shuffle phases ensure that each node will have all the intermediate values required for the Reduce phase. Note that the communication constraint in \eqref{eq:comc2} holds true for all the nodes.

{\emph{Reduce phase}}: 
Finally, in this phase, Node~$k$, $k\in [1:K]$ collects all the intermediate values $ \{a_{q,n}: q \in \mathcal{W}_k, n \in[1: N]\}$, from the Map phase and the Shuffle phase, as the inputs to compute the Reduce function.
For this proposed scheme, the total computation load is
\begin{align*}
\Mtn= 11, 
\end{align*}
which is optimal. This scheme is an example of the general scheme in Section~\ref{sec:achiK}. 

\subsection{Example with ($L_1=2, L_2=4, L_3=6$, $N=8$, $Q=K=3$)} \label{sec:asymmetric}
Now, let us consider a distributed system with heterogeneous communication load constraints and observe the minimum total computation load. The scheme design for this example with asymmetric communication load constraints ($L_1 \neq L_2 \neq L_3$) is shown in Fig.~\ref{fig:example_asymmetric}.

\begin{figure}
\centering
\includegraphics[scale=1.25, width=15cm]{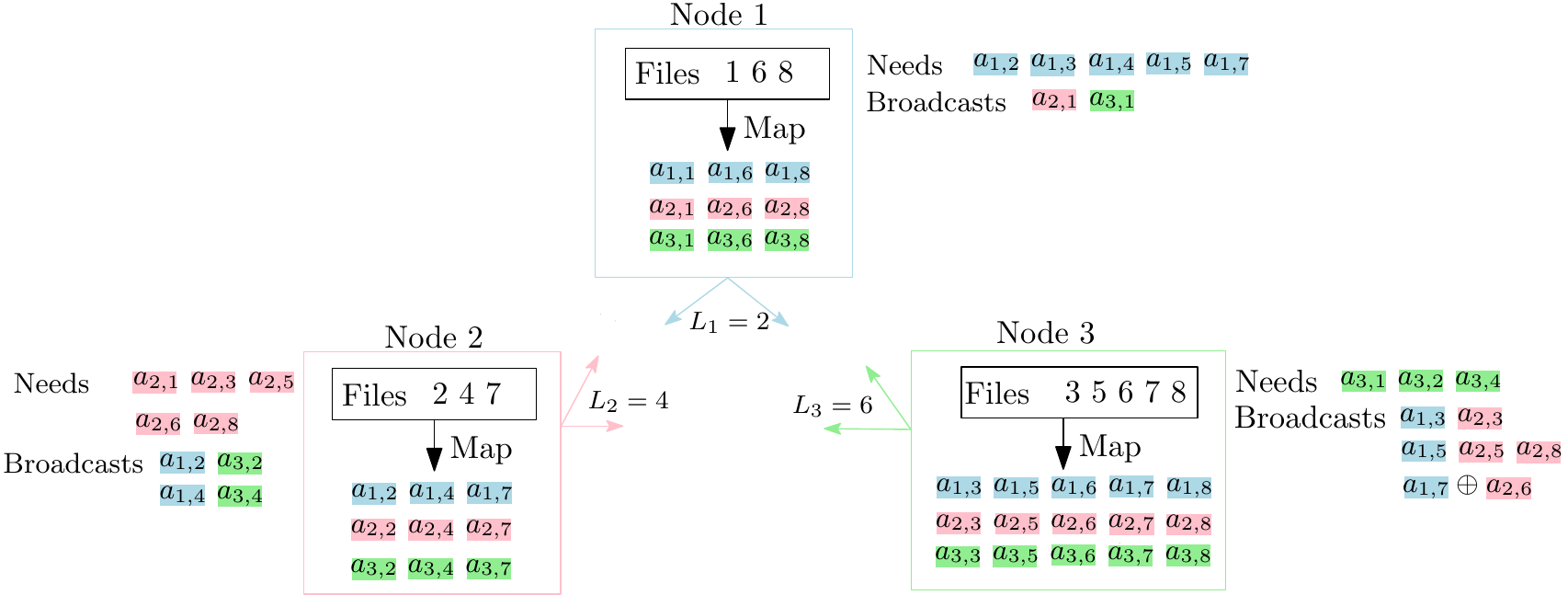}
\caption{A scheme for a distributed computing system with  asymmetric communication load constraints $( K=Q=3,  N=8, L_1=2, L_2=4, L_3 =6)$. 
The scheme is optimal in terms  of the minimum \emph{total computation load}.
 } 
\label{fig:example_asymmetric}
\end{figure}
{\emph{File placement}}:  Similar to the circular file placement procedure in Sections \ref{subsec:totalscheme1} and \ref{sec:symmetric}, file placement is started from the node with the lowest communication load which is Node~$1$ in this case. However, after placing the files $w_1, w_2$ and $w_3$ in the Nodes~$1$, $2$ and $3$, respectively, as we can see later that the communication load of Node~$1$ will reach its communication limit ($L_1=2 $). So, by bypassing the Node~$1$, file $w_4$ is placed at Node~$2$ and file $w_5$ is placed at Node~$3$. Now, due to the communication limit, the system will not place any additional exclusive files at Node~$2$. Since Node~$3$ still has enough communication capability, files $w_6, w_7$ and $w_8$ are placed at Node~$3$ and each file is also placed at one of the other nodes due to coding pattern that will be shown later. Note that in the Shuffle phase this trick of using redundant files makes coding of intermediate values possible which in turn uses less communication load. The final file placement can be designed as:
\begin{align*}
&\Mc_1 = \{1,6,8\} \text{ with }  M_1 = 3,\\
&\Mc_2 = \{2,4,7\} \text{ with }  M_2 = 3,\\
&\Mc_3 = \{3,5,6,7,8\} \text{ with }  M_3 = 5.
\end{align*}
{\emph{Map phase}}: Node~$k$ generates the following  intermediate values  
$\{a_{q,n}: q\in [1: Q], n \in  \mathcal{M}_k \}$, for  $k=1,2, 3$. 

{\emph{Shuffle phase}}: Node~$1$ sends intermediate values $\{a_{2,1}, a_{3,1}\}$, of at most $2B$ bits of information to the other nodes. Similarly, Node~$2$ sends out the message of $\{a_{1,2}, a_{3,2}, a_{1,4}, a_{3,4}\}$ carrying at most $4B$ bits to the other nodes while Node~$3$ broadcasts the message of $\{a_{1,3}, a_{2,3}, a_{1,5}, a_{2,5}, a_{2,8}, a_{1,7} \oplus a_{2,6}\}$ which carries at most $6B$ bits of information. All of the nodes satisfy the communication constraint in \eqref{eq:comc2} for this case with $L_1 = 2, L_2= 4, L_3 = 6$, that is,
\begin{align*}
\frac{K \cdot \Hen(x_{1}(a_{2,1},a_{3,1})) }{QB} \leq  \frac{2K  B }{QB} = L_1=  2,\\
\frac{K \cdot \Hen(x_{2}(a_{1,2}, a_{3,2}, a_{1,4}, a_{3,4})) }{QB} \leq  \frac{4K  B }{QB} = L_2=  4,\\
\frac{K \cdot \Hen(x_{3}(a_{1,3}, a_{2,3}, a_{1,5}, a_{2,5}, a_{2,8}, a_{1,7} \oplus a_{2,6})) }{QB} \leq  \frac{6K  B }{QB} = L_3=  6.
\end{align*}

{\emph{Reduce phase}}: 
Node~$k$, $k\in [1:K]$ has all the intermediate values $ \{a_{q,n}: q \in \mathcal{W}_k, n \in[1: N]\}$ as the inputs to compute the Reduce function.
For this proposed scheme, the total computation load is
\begin{align*}
\Mtn= 11, 
\end{align*}
which is optimal. This scheme is also an example of the general scheme in Section~\ref{sec:achiK}. 

From the examples \ref{sec:symmetric} and \ref{sec:asymmetric}, it can be observed that both the schemes have the same minimum total computation load even though they have different individual communication load constraints. We later see in the general scheme that the minimum total computation load is not affected by individual communication load constraint parameters,  $L_k \in \mathbb{N},  k \in \{1,2,\cdots,K \}$, as long as the total communication load, $L$, does not change.

\section{Achievability for the $K$-node system: The total computation load} \label{sec:achiK}
In this section, we will provide a general  scheme to achieve the minimum total computation load $\Mt$  for a  $K$-distributed computing system with heterogeneous communication load constraint parameters $L_k \in \mathbb{N},  k \in \{1,2,\cdots,K \}
$. This section provides the achievability proof for Theorems~\ref{thm:K2} and \ref{thm:K3} and Propositions~\ref{prop:KsmallL} and \ref{prop:KlargeL}. Recall that we consider the symmetric job assignment such that each node  computes $Q/K$ output functions, with $Q/K \in \mathbb{N}$. WLOG, we consider the following job assignment,
\begin{align*}
\mathcal{W}_k =[(k-1) Q/K + 1: k Q/K], \quad \quad k \in [1:K].
\end{align*}

Intuitively, one might want to reduce the computation load of the distributed system by increasing the communication load. However, simply maximizing the communication load to minimize the computation load is not an option in this setting as the  system has constraints of limited communication.
For minimizing the total computation load, we explore various strategies of delivering intermediate values among the $K$ nodes. 
The complexity of designing the optimal scheme for achieving the minimum total computation load increases as the total number of nodes, $K$, increases. 
For $K$ number of nodes, each file can be placed exclusively at $p$ nodes where $p \in [1:K]$ and we explore two strategies for sending out the intermediate values of these files during Shuffle phase: the \textit{coding strategy} and the \textit{redundancy strategy}. In \textit{coding  strategy}, intermediate values are sent out to the other nodes with $p$-order coding. The $p$-order coding refers to the bitwise XOR operation of $p$ intermediate values. In \textit{redundancy  strategy}, the intermediate values are sent out to the other nodes without coding in the Shuffle phase. We can see that with increasing $K$, there are more choices of placing the files and more strategies available in sending out the intermediate values associated with those files. It's very challenging to find the optimal scheme to achieve the minimum total computation load for a $K$-node distributed computing system when $K$ increases. With the careful choice of the possible strategies of \textit{coding} and \textit{redundancy}, we have designed our achievable scheme for a general $K$-node distributed computing system to achieve the minimum total computation load.

We introduce parameters $\Rp$ and $\Cp$ to design the scheme for the total computation load for $K$-node distributed computing system. Parameter $\Cp$ denotes the number of \textit{batches} of files placed exclusively at the nodes to ensure \textit{coding strategy} can be utilized in the Shuffle phase. The intermediate values associated with these files are broadcast with $p$-order coding in the Shuffle phase. Note that each batch of files consists of $K-1 \choose p -1$ files, $\forall p \in [1:K]$. This is because, each file from a batch is placed at one node, e.g., Node~$k, k \in [1:K]$ and at $p-1$ nodes from remaining $K-1$ nodes. In this coding strategy, for each batch of files one node is responsible for sending out ${K-1 \choose p}Q/K$ coded intermediate values to the other nodes. This is because, intermediate values of $p$-order should be sent from one node, e.g., Node~$k, k \in [1:K]$ to $K-1$ nodes. Note that for $p=K$, the files are placed at all the $K$ nodes, so, communication is not needed in the Shuffle phase. This scheme design will be explained later with exact parameter values via an example.  For coding strategy, the total computation load, denoted by $M_{\text{total}}^c$, is given as, 
\begin{align}\label{eq:totalcod}
M_{\text{total}}^c =  \sum_{p=1}^{K} p \cdot \Cp {K-1 \choose p-1},
\end{align}
and a total of
\begin{align}\label{eq:commcod}
\sum_{p=1}^{K-1} {K-1 \choose p} \Cp \cdot Q/K
\end{align}
intermediate values need to be delivered among the nodes for coding strategy, with each coded intermediate value carrying at most $B$ bits of information. 

Note that we have two extreme cases for $p=1$ and $p=K$. Parameter $p=1$ means that files are only placed at one node and the intermediate values associated with these files are sent out with $1$-order coding which implies that the broadcasted intermediate values are actually uncoded. Parameter $p=K$ suggests that the files are placed at all the $K$ nodes and hence, communication is not needed in the Shuffle phase.

The other parameter, $\Rp$, represents the number of files placed exclusively at $p$ nodes and the intermediate values associated with these files are sent out without coding in Shuffle phase, where $p \in [2:K-1]$. In this redundancy strategy, for $\Rp$ number of files, $(K-p) \Rp Q/K$ intermediate values need to be delivered to the other nodes. This is because for each file exclusively placed  at $p$ nodes, intermediate values associated with these files need to be sent out to the other $K-p$ nodes. 
This strategy will be clarified later with specific parameter design for an example.
The total computation load occupied by these $\Rp$ files with redundancy strategy, denoted by $M_{\text{total}}^r$,  is given as,
\begin{align} \label{eq:totalredun}
M_{\text{total}}^r =  \sum_{p=2}^{K-1} p \cdot \Rp,
\end{align}
and a total of
\begin{align}\label{eq:commredun}
\sum_{p=2}^{K-1} (K-p) \Rp \cdot Q/K
\end{align}
intermediate values need to be delivered among the nodes for redundancy strategy, with each intermediate value carrying at most $B$ bits of information. 

Note that for the example with ($L_1=2, L_2=4, L_3=6$, $N=8$, $Q=K=3$) in Section~\ref{sec:asymmetric} (also Fig.~\ref{fig:example_asymmetric}), the values of the above parameters are designed as
\begin{align*}
&\gamma_{c,1} = 5, \quad  \gamma_{c,2} = 1, \quad  \gamma_{c,3} = 0 \quad \text{and} \quad \gamma_{r,2} = 1.
\end{align*}
For the coding strategy, for $p=1$, $\gamma_{c,1}$ \textit{batches} of files are placed at the nodes and each \textit{batch} consists of ${K-1 \choose p-1} = 1$ file. The $\gamma_{c,1}$ \textit{batches} of files are represented by $\{w_1\},\{w_2\},\{w_3\},\{w_4\},\{w_5\}$. File $w_1$ is placed at Node~$1$, files $w_2$ and $w_4$ are placed at Node~$2$ and files $w_3$ and $w_5$ are placed at Node~$3$.  In the Shuffle phase, a total of ${K-1 \choose p} \gamma_{c,1} = 10$ intermediate values, $\{a_{2,1}, a_{3,1}, a_{1,2}, a_{3,2}, a_{1,4}, a_{3,4}, a_{1,3}, a_{2,3}, a_{1,5}, a_{2,5}\}$, are sent out by the nodes for these files. For $p=2$, $\gamma_{c,2}$ \textit{batch} of files with each \textit{batch} consisting of ${K-1 \choose p-1} = 2$ files are placed at the nodes. The files for $\gamma_{c,2}$ \textit{batch} of files are represented by $\{w_6,w_7\}$. File $w_6$ is placed at Nodes~$3$ and $1$ and file $w_7$ is placed at Nodes~$2$ and $3$. In the Shuffle phase, ${K-1 \choose p } \Cp = 1$ coded intermediate value $\{a_{1,7} \oplus a_{2,6}\}$ is broadcast by Node~$3$ to the other nodes.
For the redundancy strategy, $\gamma_{r,2}$ number of file represented by $w_8$ is placed at Nodes~$1$ and $3$. In the Shuffle phase, Node~$3$ delivers $(K - p) \gamma_{r,2} = 1$ intermediate value $\{a_{2,8}\}$ to Node~$2$. 

We design redundancy and coding strategies with the parameters $\Rp$ and $\Cp$ such that
\begin{align} 
&\sum_{p=2}^{K-1} (K-p) \Rp + \sum_{p=1}^{K-1}{K-1 \choose p}\Cp \leq L,  \label{eq:gencomm}\\
&\sum_{p=2}^{K-1} \Rp + \sum_{p=1}^{K} \Cp {K-1 \choose p-1} = N. \label{eq:genfile}
\end{align}
The condition in \eqref{eq:genfile} guarantees that, with our designed parameters $\{\Rp,\Cp\}_p$, all the $N$ input files are placed at the $K$ distributed nodes. The condition in \eqref{eq:gencomm} guarantees that our scheme design satisfies the following communication load constraint (see \eqref{eq:comc2} and \eqref{eq:sumL})
\[\frac{K\cdot \sum_{k=1}^K \Hen(x_k)}{QB} \leq \sum_{p=2}^{K-1} (K-p) \Rp + \sum_{p=1}^{K-1}{K-1 \choose p}\Cp \leq L.\]
For the proposed scheme, the file placement, Map, Shuffle, and Reduce phases are explained in Algorithm~\ref{generalAlg}, which takes the designed parameters satisfying conditions~\eqref{eq:gencomm} and \eqref{eq:genfile} as inputs. Note that Algorithm~\ref{generalAlgFunc} provides some functions needed for Algorithm~\ref{generalAlg}.
In the algorithms that follow, we use the following notations
\begin{align}
x^\dagger &\defeq [(x-1) \text{ mod } K ]+1, \quad \text{for} \quad x \in \mathbb{N}^+.  \non
\end{align}

\begin{algorithm}
\caption{Function definitions for Algorithm \ref{generalAlg}} \label{generalAlgFunc}
\begin{algorithmic}[1]
\Function{nextnodecoding}{$k, load,p,i$} 
	\While $(L_k < load)$
		\State $k \gets (k + 1)^\dagger$

	\EndWhile
	\State\Return $k$
\EndFunction

\Function{nextnoderedundancy}{$k,p,i$}
	\While $\exists \Ac:\Ac \subseteq [1:K], |\Ac| = p, k \in \Ac, \sum\limits_{j:j \in \Ac} L_j < K-p$
		\State $k \gets \ksum$ 
	\EndWhile	
	\State $\bar{\Ac}\gets $ Select one $\Ac$
	\State\Return $(\bar{\Ac},  k)$
\EndFunction

\Function{removenode}{$\hat{\Ac},k$}
	\While $L_k =0$
	\State $\hat{\Ac} \gets \hat{\Ac} \setminus \{k\}$
    \State $\text{Choose largest } \hat{k} \in \hat{\Ac}$
	\EndWhile
	\State \Return $(\hat{\Ac},\hat{k})$
\EndFunction
\algstore{bkbreak}
\end{algorithmic}
\end{algorithm}

The scheme design for the total computation load for a system with $K$ distributed nodes consists of redundancy and coding strategies, and the following total computation load for the proposed scheme is achievable:
\begin{align}
\Mtn &= \min \{M_{\text{total}}^r + M_{\text{total}}^c\},  \label{eq:gentotal}\\
&\text{s.t.} \ \ \text{conditions \eqref{eq:gencomm} and \eqref{eq:genfile} are satisfied}. \non
\end{align}
From \eqref{eq:totalcod} and \eqref{eq:totalredun}, the total computation load in \eqref{eq:gentotal} can be rewritten as,
\begin{align}
\Mtn &=  \min \Big\{ {\sum_{p=2}^{K-1} p \cdot \Rp +  \sum_{p=1}^{K} p \cdot \Cp {K-1 \choose p-1}} \Big\},\\
&\text{s.t.} \ \ \text{conditions \eqref{eq:gencomm} and \eqref{eq:genfile} are satisfied} \non.
\end{align}

\begin{algorithm}[h!]
\caption{Achievable Scheme for $\Mt$ with ($Q/K \in \mathbb{N}^+$, $L_k/(K-1) \in \mathbb{N}$, $k=1,2, \cdots, K $)}\label{generalAlg}
\begin{algorithmic}[1]
\algrestore{bkbreak}
\Procedure{File Assignment}{}
\State $k\gets1$; \ $n\gets 1;$ \ $\hat{L_{\ell}} \gets L_{\ell}, \ell =1,2, \cdots, K$

\For{$p =1:K$}
	\For{$i = 1: \Cp$}
		\State $k \gets$ \Call{nextnodecoding}{$k,{K-1 \choose p}, p, i$}; \quad 	 $k_{i,p}^c \gets k$ 
		\For{\textbf{each} $\Ac:\Ac \subseteq [1:K], |\Ac| = p, k \in \Ac $}
			\State  Place file $w_n$ at the Nodes indexed by $\Ac$; \ \ $n \gets n +1$
		\EndFor
		\State $L_k \gets L_k - {K-1 \choose p}$; \quad $k \gets \ksum$
	\EndFor
\EndFor
\For{$p=2:K-1$} 
\For{$i = 1: \Rp$}
		\State $(\Ac,k) \gets$  \Call{nextnoderedundancy}{$k,p,i$};  \quad  $k_{i,p}^r \gets k$; \quad $\bar{\Ac}_{i,p} \gets \Ac$
		\State Place file $w_n$ at the Nodes indexed by $\Ac; \quad \hat{\Ac} \gets \Ac$
\For{$j \in [1:K] \setminus \Ac$}
	\State $(\hat{\Ac},k) \gets $ \Call{removenode}{$\hat{\Ac},k$}
	\State  $L_k \gets L_k - 1$
\EndFor
\State $n \gets n + 1; \quad k \gets \ksum$

\EndFor
\EndFor
\EndProcedure

\Procedure{Map phase}{}
	\State  Node $k$ computes \emph{Map} functions and outputs $\big\{a_{q,n}$: $q\in [1:Q]$, $n \in$ $\mathcal{M}_k$; \ $k = 1,2, \cdots,K \big\}$

\EndProcedure

\Procedure{Shuffle phase}{}
\State  $n\gets 1;$ \ $ L_{\ell} \gets \hat{L_{\ell}} , \ell = 1,2, \cdots, K $

\For{$p =1:K$}
	\For{$i = 1: \Cp$}
		\State $k \gets k_{i,p}^c$; $\quad \Bc \gets \Big[n: n + {K-1 \choose p-1}-1\Big]$ 
		\For{\textbf{each} $\Ac:\Ac \subseteq [1:K], |\Ac| = p, k \notin \Ac $}
			\State  Node $k$ broadcasts intermediate value $\Big\{\bigoplus\limits_{\forall j \in \Ac} a_{\Wc_{j},     \{\cap_{ \ell \in   \{k\} \cup \Ac \setminus \{j\}}   \Mc_{ \ell} \}   \cap \Bc  }\Big\}$
		\EndFor
		\State $n \gets n + {K-1 \choose p-1};$ \quad $L_k \gets L_k - {K-1 \choose p}$
	\EndFor
\EndFor
\For{$p=2:K-1$}
\For{$i = 1: \Rp$}
	\State $\hat{\Ac} \gets \bar{\Ac}_{i,p};$ \quad  $k \gets k_{i,p}^r$
	    \For{$ j \in [1:K] \setminus \bar{\Ac}_{i,p}$} 
		\State $(\hat{\Ac},k) \gets $ \Call{removenode}{$\hat{\Ac},k$}
			\State Node~$k$ sends $a_{\Wc_{j},n}$ to Node $j$; \ $L_k \gets L_k - 1$; 
		\EndFor
	\State $n \gets n + 1$

\EndFor
\EndFor
\EndProcedure

\Procedure{Reduce phase}{}
	\State Node $k$ computes \emph{Reduce} functions indexed by $\Wc_k$; \ \ $k =  1,2, \cdots,K $

\EndProcedure

\end{algorithmic}
\end{algorithm}

In the following, we have computed the parameters $\Rp$ and $\Cp$ specifically to achieve the optimal (minimum) total computation load, given in Theorems~\ref{thm:K2} and \ref{thm:K3} and Propositions~\ref{prop:KsmallL} and \ref{prop:KlargeL}, for a distributed computing system.

\subsection{Parameter design for Theorem \ref{thm:K2}} \label{subsec:thm2}
For this setting with $K=2$, we design the parameters as  
\begin{align*} 
\gamma_{c,1} &= \min\{L,N\}, \\
\gamma_{c,2} &= (N - L)^+.
\end{align*}
One can check that the above design of the parameters satisfies the conditions~\eqref{eq:gencomm} and \eqref{eq:genfile}. Then, the total computation load for a distributed computing system with $K=2$ is given by:
\begin{align} 
\Mtn &=  \gamma_{c,1}+ 2\gamma_{c,2} \\
&=  \max \{N, 2N-L\},
\end{align}
which turns out to be optimal (see Theorem \ref{thm:K2}).

\subsection{Parameter design for Theorem \ref{thm:K3}} \label{subsec:thm3}
For a general setting of a three-node system, $K=3$, we design the parameters of the proposed scheme as
\begin{align*}
\gamma_{c,2}&= 3N - \min\{L,2N\} - \max\Big\{N, \ceil*{\frac{7N-2L}{3}}, 3N-2L\Big\},\\
\gamma_{c,1} &= \min\{L,2N\} + \gamma_{c,2} - \min\{2L, N\},\\
\gamma_{r,2} &= 2\min\{2L,N\} - 3\gamma_{c,2} - \min\{L,2N\}, \\
\gamma_{c,3} &= N  - \gamma_{r,2} - \gamma_{c,1} - 2\gamma_{c,2}.
\end{align*}

One can check that the constraints in \eqref{eq:gencomm} and \eqref{eq:genfile} are satisfied with this design. 
With the above parameter design, from \eqref{eq:gentotal}, the total computation load of  the proposed scheme is given by 
\begin{align}  
\Mtn &= 2\gamma_{r,2} +   \gamma_{c,1}+ 4\gamma_{c,2} + 3\gamma_{c,3}   \label{eq:totallast11}  \\ 
&=  \max \Bigl\{N, \ceil*{\frac{7N}{3}-\frac{2L}{3}}, 3N-2L\Bigr\}. \label{eq:totallast22}
\end{align}

Note that, for $K=3$, the total computation load of the proposed scheme turns out to be optimal when $L_k/2 \in \mathbb{N}$.

\subsection{Parameter design for Proposition \ref{prop:KsmallL}} \label{subsec:coding}
Let us consider the case when $L$ is very small compared to given $N$. More specifically we look at the case of

\[L \leq \frac{N}{K-1}.\]
In this case, the parameters are set as
\begin{align*}
&\Cp =\begin{cases} 0, \quad & p = 1,2, \cdots,K-2,\\
 L, \quad  & p = K-1,\\
 N - (K-1)L, \quad & p = K, 
 \end{cases}\\
&\Rp = 0, \quad \quad \quad \quad \quad \quad \quad \quad \ p = 2,3, \cdots,K-1.
\end{align*}
One can easily verify that these parameters satisfy both the constraints \eqref{eq:gencomm} and \eqref{eq:genfile}. Then, the total computation load from \eqref{eq:gentotal} is given by:
\begin{align}
\Mtn &= (K-1)L {K-1 \choose K-2} + K[N-(K-1)L] {K-1 \choose K - 1}\\
&= KN - (K-1)L,
\end{align}
which turns out to be optimal (see Proposition \ref{prop:KsmallL}).

\subsection{Parameter design for Proposition \ref{prop:KlargeL}} \label{subsec:nocoding}
In this case, we consider the case of
\[\min_k \{L_k\} \geq (K-1) \cdot \ceil*{N/K},\]
and then, we set the parameters as
\begin{align*}
&\Cp =\begin{cases} N , \quad & p = 1,\\
 0, \quad  & p = 2,3, \cdots,K,\\
 \end{cases}\\
&\Rp = 0, \quad \quad \quad p = 2,3, \cdots,K-1.
\end{align*}
These parameters satisfy the constraints \eqref{eq:gencomm} and \eqref{eq:genfile} and the total computation load from \eqref{eq:gentotal} is given by
\begin{align}
\Mtn &= N,
\end{align}
which turns out to be optimal (see Proposition \ref{prop:KlargeL}).
\vspace{3pt}

\section{Achievability for the three-node system: The worst-case computation load}  \label{sec:achiMwt3}
In this section, we design a general  scheme by focusing on the worst-case computation load, for a three-node ($K=3$) distributed computing system.
For some cases, the proposed   scheme achieves the optimal (minimum) worst-case computation load  (as shown in Proposition~\ref{prop:K3optimal}), by setting the scheme parameters specifically. 
  In the proposed scheme,  at first $N$ input files are divided into 16 disjointed groups, whose indices are given by
\begin{align}  \label{eq:groups}
 \Sc_{\{1\}}, \Sc_{\{2\}},  \Sc_{\{3\}},  \Sc_{\{1, 2\}}^{r_1},  \Sc_{\{1, 2\}}^{r_2}, \Sc_{\{1, 2\}}^{c_1}, \Sc_{\{1, 2\}}^{c_2} , \Sc_{\{1, 3\}}^{r_1} ,  \Sc_{\{1, 3\}}^{r_3},  \Sc_{\{1, 3\}}^{c_1},  \Sc_{\{1, 3\}}^{c_3}   ,   \Sc_{\{2, 3\}}^{r_2} ,  \Sc_{\{2, 3\}}^{r_3},  \Sc_{\{2, 3\}}^{c_2} , \Sc_{\{2, 3\}}^{c_3}  ,  \Sc_{\{1, 2, 3\}},
\end{align}  
where $\Sc_{\Ac}$, $\Ac \subseteq [1:3]$,  has been defined in \eqref{eq:Sdefine}, and the other notations are defined as follows.  
For any $ \Ac  \subseteq [1:3]$ with $|\Ac|=2$, $\Sc_{\Ac}^{c_j}$ and  $\Sc_{\Ac}^{r_j}$ denote the sets of indices of files placed exclusively at the nodes indexed by $\Ac$,  and Node $j, j \in \Ac$, is responsible for sending out the required intermediate values associated with these files  in the Shuffle phase.  For the files $w_n, n\in \Sc_{\Ac}^{c_j}$,   we  employ the \emph{coding strategy}, which has been used in the scheme in Section~\ref{sec:achiK}. 
Similarly, for the files $w_n, n\in \Sc_{\Ac}^{r_j}$,   we  employ the \emph{redundancy strategy}.

In our scheme we have $16$ parameters 
\begin{align}  \label{eq:groupssizes}
S_1,\  S_2, \  S_3,  \  S_{\Ac}^{r_j},  \ S_{\Ac}^{c_j}, \ S_{123}, \quad     \ \Ac= \{1, 2\}, \{1, 3\}, \{2, 3\}, \  j \in \Ac,
\end{align} 
where $S_{{\Ac}}^{r_j} \defeq |\Sc_{{\Ac}}^{r_j}|$, $S_{{\Ac}}^{c_j} \defeq |\Sc_{{\Ac}}^{c_j}|$,  $S_{i} \defeq |\Sc_{\{i\}}|$, $ i=1,2,3,$ and $S_{123}\defeq |\Sc_{\{1,2,3\}}|$ (see \eqref{eq:Sdefine}).
For simplicity  we use  the notation $S_{12}^{r_1}$ to represent $S_{\{1,2\}}^{r_1} $ and similar notations are used later on.  
For the proposed scheme we design these $16$ parameters under the following conditions:
\begin{align}
&\sum_{i=1}^{3}S_i+ \sum_{\Ac : \Ac \subseteq [1: 3], |\Ac|=2} \  \sum_{ j \in \Ac} (S_{\Ac}^{r_j}+S_{\Ac}^{c_j}) + S_{123} = N,  \label{eq:fileprocess} \\ 
&S_{12}^{c_1}=S_{13}^{c_1},  \quad  S_{12}^{c_2}=S_{23}^{c_2}, \quad  S_{13}^{c_3}=S_{23}^{c_3},  \label{eq:coneq22} \\ 
&2S_1 + S_{12}^{r_1}+ S_{13}^{r_1} + \frac{1}{2}(S_{12}^{c_1}+S_{13}^{c_1}) \leq L_1,   \label{eq:comma} \\
&2S_2 + S_{12}^{r_2} + S_{23}^{r_2} + \frac{1}{2}(S_{12}^{c_2}+S_{23}^{c_2}) \leq L_2,   \label{eq:commb} \\
& 2S_3 + S_{13}^{r_3} + S_{23}^{r_3} + \frac{1}{2}(S_{13}^{c_3}+S_{23}^{c_3}) \leq L_3 .     \label{eq:commc}
\end{align}
The condition in \eqref{eq:fileprocess} guarantees that the 16 disjointed groups in \eqref{eq:groups} should include all the indices of $N$ input files. 
The condition in \eqref{eq:coneq22} guarantees that the scheme can utilize the \emph{coding strategy} used in  Section~\ref{sec:achiK}. 
The conditions in \eqref{eq:comma}-\eqref{eq:commc} guarantee that the communication load constraint in \eqref{eq:comc2} is satisfied.
Let us focus on Node~1 and consider the Shuffle phase.   For each file index in group $\Sc_{\{1\}}$, $2Q/K$ intermediate values need to be delivered; for every \emph{two} file indices in group $\Sc_{\{1, 2\}}^{c_1} \cup \Sc_{\{1, 3\}}^{c_1}$,  $Q/K$ \emph{coded} intermediate values need to be delivered; for each file index in $\Sc_{\{1, 2\}}^{r_1} \cup \Sc_{\{1, 3\}}^{r_1}$, $Q/K$ intermediate values need to be delivered;  and for all the file indices in group $\Sc_{\{1, 2, 3\}}$, no communication is required. 
Therefore, a total of  \[\frac{2S_1 + S_{12}^{r_1} + S_{13}^{r_1} + \frac{1}{2}(S_{12}^{c_1}+S_{13}^{c_1})}{K/Q}\]  (coded and uncoded) intermediate values need to be delivered from Node~1, each intermediate value carrying at most $B$ bits of information.  The condition in \eqref{eq:comma}  guarantees that the  scheme design satisfies the communication load constraint for Node~1, that is  \[   \frac{K \cdot \Hen(x_{1}) }{QB}  \leq   \frac{KB\cdot \bigl(2S_1 + S_{12}^{r_1}+ S_{13}^{r_1} + \frac{1}{2}(S_{12}^{c_1}+S_{13}^{c_1})\bigr)Q/K}{QB} =  2S_1 + S_{12}^{r_1}+ S_{13}^{r_1} + \frac{1}{2}(S_{12}^{c_1}+S_{13}^{c_1}) \leq L_1.\] 
Similarly, the conditions in \eqref{eq:commb} and  \eqref{eq:commc}   guarantee that the  scheme design satisfies the communication load constraints for Node~2 and Node~3, respectively.
The file placement, Map, Shuffle, and Reduce phases of the proposed scheme are explained in Algorithm~\ref{generalAlgworst}, which takes any 16 parameters (see \eqref{eq:groupssizes}) satisfying conditions~\eqref{eq:fileprocess}-\eqref{eq:commc} as inputs.

\begin{algorithm}
\caption{Achievable Scheme for $\Mwt$ with ($K=3$, $Q/K \in \mathbb{N}^+$)}\label{generalAlgworst}
\begin{algorithmic}[1]
\Procedure{Initialization}{}
\State $z \gets 0;$\quad  $\uv \gets [S_{1}, S_{2},  S_{3},  S_{12}^{r_1},  S_{12}^{r_2}, S_{12}^{c_1}, S_{12}^{c_2} , S_{13}^{r_1} ,  S_{13}^{r_3},  S_{13}^{c_1},  S_{13}^{c_3}   ,   S_{23}^{r_2} ,  S_{23}^{r_3},  S_{23}^{c_2} , S_{23}^{c_3}  ,  S_{123}]^\T$
\For{$i = 1 : 16$}
	\State $\Cc_i \gets [z+1:z +\uv[i]];$ \quad $z \gets z + \uv [i]$
\EndFor
\State $\Sc_1 \gets \Cc_1; \quad \Sc_2 \gets \Cc_2; \quad  \Sc_3 \gets \Cc_3; \quad \Sc_{\{1, 2\}}^{r_1} \gets \Cc_4; \quad  \Sc_{\{1, 2\}}^{r_2} \gets \Cc_5; \quad  \Sc_{\{1, 2\}}^{c_1}\gets \Cc_6;$
\State $\Sc_{\{1, 2\}}^{c_2} \gets \Cc_7; \quad \Sc_{\{1, 3\}}^{r_1} \gets \Cc_8;\quad \Sc_{\{1, 3\}}^{r_3}\gets \Cc_9; \quad \Sc_{\{1, 3\}}^{c_1} \gets \Cc_{10}; \quad \Sc_{\{1, 3\}}^{c_3}  \gets \Cc_{11};$
\State $\Sc_{\{2, 3\}}^{r_2}\gets \Cc_{12}; \quad \Sc_{\{2, 3\}}^{r_3} \gets \Cc_{13}; \quad \Sc_{\{2, 3\}}^{c_2} \gets \Cc_{14}; \quad  \Sc_{\{2, 3\}}^{c_3}\gets \Cc_{15}; \quad \Sc_{\{1, 2, 3\}} \gets \Cc_{16}$
\EndProcedure
\Procedure{File Assignment}{}
\State $\Mc_1 \gets \Sc_1\cup \Sc_{\{1,2\}}^{r_1}\cup \Sc_{\{1,2\}}^{r_2}\cup \Sc_{\{1,2\}}^{c_1}\cup \Sc_{\{1,2\}}^{c_2}\cup \Sc_{\{1,3\}}^{r_1}\cup \Sc_{\{1,3\}}^{r_3}\cup \Sc_{\{1,3\}}^{c_1}\cup \Sc_{\{1,3\}}^{c_3}\cup \Sc_{\{1,2,3\}}$
\State $\Mc_2 \gets \Sc_2\cup \Sc_{\{1,2\}}^{r_1}\cup \Sc_{\{1,2\}}^{r_2}\cup \Sc_{\{1,2\}}^{c_1}\cup \Sc_{\{1,2\}}^{c_2}\cup \Sc_{\{2,3\}}^{r_2}\cup \Sc_{\{2,3\}}^{r_3}\cup \Sc_{\{2,3\}}^{c_2}\cup \Sc_{\{2,3\}}^{c_3}\cup \Sc_{\{1,2,3\}}$
\State $\Mc_3 \gets \Sc_3\cup \Sc_{\{1,3\}}^{r_1}\cup \Sc_{\{1,3\}}^{r_3}\cup \Sc_{\{1,3\}}^{c_1}\cup \Sc_{\{1,3\}}^{c_3}\cup \Sc_{\{2,3\}}^{r_2}\cup \Sc_{\{2,3\}}^{r_3}\cup \Sc_{\{2,3\}}^{c_2}\cup \Sc_{\{2,3\}}^{c_3}\cup \Sc_{\{1,2,3\}}$
\For{ $k =  1: K $ }
	\State  Place all the files indexed by $\Mc_k$ at Node $k$
\EndFor
\EndProcedure

\Procedure{Map phase}{}
\For{ $k =  1: K $ }
	\State  Node $k$ computes \emph{Map} functions and outputs $a_{q,n}$, $q\in [1:Q]$ and $n \in \mathcal{M}_k$
\EndFor
\EndProcedure

\Procedure{Shuffle Phase}{}
\For{$k = 1:K$}
	\State Node~$k$ sends $a_{\Wc_{\ksum},\Sc_k}$ to Node~$\ksum$ and $a_{\Wc_{\ksumm},\Sc_k}$ to Node~$\ksumm$
	\State Node~$k$ sends $a_{\Wc_{\ksumm},\Sc_{\{k, \ksum\}}^{r_k}}$ to Node~$\ksumm$ and $a_{\Wc_{\ksum},\Sc_{\{k, \ksumm\}}^{r_k}}$ to Node~$\ksum$
	\State Node~$k$ broadcasts $a_{\Wc_{\ksumm},\Sc_{\{k, \ksum\}}^{c_k}} \oplus a_{\Wc_{\ksum},\Sc_{\{k, \ksumm\}}^{c_k}}$ to the other nodes
\EndFor

\EndProcedure
\Procedure{Reduce phase}{}
\For{ $k =  1: K $ }
	\State Node $k$ computes \emph{Reduce} functions indexed by $\Wc_k$
\EndFor
\EndProcedure
\end{algorithmic}
\end{algorithm}

For our proposed scheme, by optimizing over the 16 parameters satisfying conditions \eqref{eq:fileprocess}-\eqref{eq:commc}, then the following worst-case computation load is achievable
\begin{align} 
\Mwtn  =    \min \  
& \max \Bigg\{  \ S_1 +\  \sum_{\Ac = \{1, 2\}, \{1, 3\}} \ \sum_{ j \in \Ac} (S_{\Ac}^{r_j}+S_{\Ac}^{c_j})+S_{123}, \non \\ 
&\ \ \ \ \qquad S_2 + \  \sum_{\Ac = \{1, 2\}, \{2, 3\}} \  \sum_{ j \in \Ac} (S_{\Ac}^{r_j}+S_{\Ac}^{c_j})+ S_{123}, \non \\ 
&\ \ \ \ \qquad S_3 +\   \sum_{\Ac = \{1, 3\}, \{2, 3\}} \ \sum_{ j \in \Ac} (S_{\Ac}^{r_j}+S_{\Ac}^{c_j})+ S_{123}\Bigg\} \label{eq:worstcase}  \\
&   \non \\
	  \text{s.t.}  \ \ &  \text{conditions~\eqref{eq:fileprocess}-\eqref{eq:commc}},  \  S_1, S_2, S_3,  S_{\Ac}^{r_j}, S_{\Ac}^{c_j}, S_{123} \in \mathbb{N}, \  \forall \Ac  \subseteq [1:3], |\Ac|=2, \forall j \in \Ac.  \non  
\end{align}

One can follow Algorithm~\ref{generalAlgworst}  to get the scheme in Fig.~\ref{egcontradict0} for given parameters $( K=Q=3,  N=7, L_1=L_2=2, L_3 =14)$ by setting  design parameters $S_1 = 1, S_3 = 4, S_{12}^{r_2} = 2$ and the rest of the design parameters as $0$. From the algorithm, we get $\Sc_1 = \{1\}, \Sc_3 = \{2,3,4,5\}, \Sc_{\{1,2\}}^{r_2} = \{6,7\}$. The design parameters satisfy the conditions~\eqref{eq:fileprocess}-\eqref{eq:commc} and the worst-case computation load from \eqref{eq:worstcase} is $4$, which turns out to be optimal.

In the following, we will set the 16 parameters specifically to achieve the optimal (minimum) worst-case computation load for some cases (as shown in Proposition~\ref{prop:K3optimal}). 

\subsection{Parameter design for Condition 1 in Proposition~\ref{prop:K3optimal}} \label{prop:convbound2proof} 

Let us consider the  case of \[\min_{k} \{L_k\} \geq 2\ceil*{\frac{N}{3}}. \]
In this case, we set the parameters as
\begin{align*}
S_1 = & \ceil*{N/3}, \quad    S_2 =   \floor*{N/3},    \quad  S_3 =  N-  \ceil*{N/3} -  \floor*{N/3}, \\
S_{\Ac}^{r_j}= &S_{\Ac}^{c_j}=S_{123}=0,  \ \  \Ac = \{1,2\}, \{1,3\}, \{2,3\},  \forall j \in \Ac.
\end{align*}
With the choices of these parameters, one can easily verify that the conditions~\eqref{eq:fileprocess}-\eqref{eq:commc} are satisfied. Then, the worst-case computation load is given by
\begin{align}
\Mwtn =&\max \Bigg\{ \ S_1 +\  \sum_{\Ac = \{1, 2\}, \{1, 3\}} \ \sum_{ j \in \Ac} (S_{\Ac}^{r_j}+S_{\Ac}^{c_j})+S_{123}, \non \\ 
&\ \  \ \ \quad \quad  S_2 + \ \sum_{\Ac = \{1, 2\}, \{2, 3\}} \  \sum_{ j \in \Ac} (S_{\Ac}^{r_j}+S_{\Ac}^{c_j})+ S_{123}, \non \\ 
&\ \ \ \ \quad \quad S_3 +\   \sum_{\Ac = \{1, 3\}, \{2, 3\}} \ \sum_{ j \in \Ac} (S_{\Ac}^{r_j}+S_{\Ac}^{c_j})+ S_{123}\Bigg\} \label{eq:worstref} \\
= & \ceil*{N/3},  \non
\end{align}
which turns out to be optimal (see \eqref{eq:K3worst1} in Proposition~\ref{prop:K3optimal}).

\subsection{Parameter design for Condition 2 in Proposition~\ref{prop:K3optimal}} \label{prop:convbound4proof} 

Now we consider the case of   \[L \leq N/2, \] 
and then we set the parameters as 
\begin{align*} 
&S_i =0, \ \  \forall i \in [1:3], \non \\
&S_{\Ac}^{r_j}=0,  \  \ \Ac = \{1,2\}, \{1,3\}, \{2,3\}, \forall j \in \Ac, \non \\
&S_{12}^{c_1}=S_{13}^{c_1}=L_1,  \non \\
&S_{12}^{c_2}=S_{23}^{c_2}=L_2, \non \\
&S_{13}^{c_3}=S_{23}^{c_3}=L_3, \non \\
&S_{123}=N-2L .
\end{align*}
One can  verify that the conditions~\eqref{eq:fileprocess}-\eqref{eq:commc} are satisfied with the choices of the parameters. Then, the worst-case computation load  (see \eqref{eq:worstref}) is given by
\begin{align}
\Mwtn =  N - \min\{L_2 + L_3, L_1+L_3, L_1 + L_2\}, \non
\end{align}
which turns out to be optimal (see \eqref{eq:K3worst2} in Proposition~\ref{prop:K3optimal}).

\subsection{Parameter design for Condition 3 in Proposition~\ref{prop:K3optimal}}\label{prop:convbound1proof}

We then consider the case with the following conditions: 
\begin{align}
2&\leq \min_k\{L_k\} \leq\frac{2N}{3},  \label{eq:cond1for3} \\
\frac{3N \!\!-\!  \min_{i\neq j}\{L_i + L_j\} }{5} & \leq \!   \ceil*{\frac{ N}{2} - \frac{\min_k\{L_k\}}{4}} \! \leq \frac{\max_k\{L_k\}}{2}.\label{eq:cond2for3}
\end{align}
WLOG, we  assume that $L_1 \leq L_2 \leq L_3$. 
Then, we set the parameters as
\begin{align*}
&S_1 = S_2 = N-2\ceil*{\frac{N}{2}-\frac{L_1}{4}},  \\
&S_3 = \ceil*{\frac{N}{2}-\frac{L_1}{4}}, \\ 
&S_{12}^{r_1}=  L_1 - 2N + 4   \ceil*{\frac{N}{2}-\frac{L_1}{4}}, \\
&S_{12}^{r_2} = N- L_1- \ceil*{\frac{N}{2}-\frac{L_1}{4}},  \\
&S_{\Ac}^{r_j} =0, \  \Ac = \{1, 3\}, \{2, 3\}, \ \forall j \in \Ac,\\ 
&S_{\Ac}^{c_j}=0, \   \Ac =  \{1, 2\}, \{1, 3\}, \{2, 3\}, \  \forall j \in \Ac, \\ 
&S_{123}= 0.
\end{align*}
One can check that conditions~\eqref{eq:fileprocess}-\eqref{eq:commc}  hold true for the choices of the parameters. Finally, the worst-case computation load is given by
\[\Mwtn = \ceil*{\frac{N}{2} - \frac{\min_k\{L_k\}}{4}},\]
which turns out to be optimal (see \eqref{eq:K3worst3} in Proposition~\ref{prop:K3optimal}). 
\vspace{3pt}

\section{Achievability for the $K$-node system: both the Total computation load and the Worst-case computation load} \label{sec:achiKsame}
Here, in this section, we focus on designing the general scheme to achieve both the worst-case computation load and the total computation load for $K$ distributed nodes, under the heterogeneous communication load constraint parameters, $L_k, k \in [1:K]$. This proposed scheme can achieve both the minimum total computation load and the minimum worst-case computation load for some cases by setting the parameters specifically which is given later on. First, we divide the $N$ files into $\sum_{p=1}^{K-1} p  \binom{K}{p} + 1$ disjointed groups whose indices are given by
\begin{align}
\Sc_{\{1 \cdots K\}}, \ \ \Sc_{\Ac}^{c_j} \ \ \forall \Ac \subseteq [1:K],  |\Ac| = p,  p=1,2,\cdots,K-1, j \in \Ac. \non	 
 \end{align}
In the above, $\Sc_{\{1 \cdots K\}}$ denotes the indices of files placed at all the $K$ nodes. $\Sc_{\Ac}^{c_j}$ denotes the indices of files exclusively placed at the nodes indexed by $\Ac$ and Node $j, j \in \Ac$, is responsible for sending out the required intermediate values associated with these files in the Shuffle phase. \textit{Coding strategy} (See Section \ref{sec:achiMwt3})  will be used to send out the intermediate values associated with files $w_n, n \in \Sc_{\Ac}^{c_j}$.  Note that for $|\Ac| = p=1$, the files are placed exclusively at one node only, so, the intermediate values associated with these files are sent out with $1$-order coding which implies that the broadcasted intermediate values are uncoded. In this scheme, the number of parameters depend upon the total number of nodes, $K$. Specifically, for given $K$, we will have $\sum_{p=1}^{K-1} p  \binom{K}{p} + 1$ parameters which are represented as
\begin{align}
S_{1 \cdots K}, \ \ S_{\Ac}^{c_j} \ \  \forall \Ac \subseteq [1:K],  |\Ac| = p,  \ p=1,2,\cdots, K-1, j \in \Ac, \label{eq:parameters}	
\end{align}
where, $S_{\Ac}^{c_j} \defeq |\Sc_{\Ac}^{c_j}|, S_{1\cdots K} \defeq | \Sc_{\{1\cdots K\}}|$  and, these parameters are designed under the following conditions:

\begin{align}
&\sum_{p=1}^{K-1} \sum_{j, \Ac:\Ac \subseteq [1:K], |\Ac|=p, j \in \Ac}  S_\Ac^{c_j} + S_{1\cdots K} = N,\label{eq:gencond1}\\
&S_{\Ac}^{c_j} = S_{\Ac^{\prime}}^{c_j}; \ \ \Ac, \Ac^{\prime} \subseteq [1:K], |\Ac| =|\Ac^{\prime}| = p, j \in \Ac \cap \Ac^\prime, \text{ for } j = 1,2, \cdots, K, \  p =2,3, \cdots, K-1, \label{eq:coding}\\
&\sum_{p=1}^{K-1}  \frac{K-p}{p} \sum_{\Ac:\Ac \subseteq [1:K], |\Ac| = p,  j \in \Ac }S_\Ac^{c_j} \leq L_j, \  \quad \quad \quad \text{ for }  j = 1,2, \cdots, K.\label{eq:gencond2}
\end{align}
The condition in \eqref{eq:gencond1} guarantees that all the indices of the $N$ input files are included in the  $\sum_{p=1}^{K-1} p  \binom{K}{p} + 1$ disjointed groups. The condition in \eqref{eq:coding} guarantees that the scheme can utilize the \emph{coding strategy} mentioned in Section~\ref{sec:achiMwt3} (also see Fig.~\ref{fig:general_example}) and the condition in \eqref{eq:gencond2} guarantees that the communication load constraint in \eqref{eq:comc2} is satisfied. 

Let us focus on Node~$1$ and consider the Shuffle phase. For \textit{coding strategy}, we consider the Shuffle phase for a group of ${K-1 \choose p-1}$ files as specified in constraint \eqref{eq:coding}. For every ${K-1 \choose p-1}$ file indices in group $\cup_{1 \in \Ac, |\Ac| = p, \Ac \subseteq [1:K]} \Sc_{\Ac}^{c_1}$,  ${K-1 \choose p}Q/K$ \emph{coded} intermediate values need to be delivered for $p \in [1:K-1]$. This is because the intermediate values, associated with files $w_n, n \in \Sc_{\Ac}^{c_1}$, are sent using $p$-order coding from node $1$ to the other $K-1$ nodes. Note that $p$-order coding refers to XOR bitwise operation of $p$-intermediate values. Also note that for $p=1$, ${K-1 \choose 1} = K-1$ uncoded intermediate values are sent from node $1$ to the other $K-1$ nodes. However, for the file indices in group $\Sc_{\{1\cdots K\}}$, no communication is required. 
Therefore, a total of  \[\sum_{p=1}^{K-1} \frac{{K-1 \choose p}}{{K-1 \choose p-1}} \sum_{\Ac:\Ac\subseteq [1:K], |\Ac|=p, 1 \in \Ac}S_{\Ac}^{c_1}Q/K = \sum_{p=1}^{K-1} \frac{K-p}{p} \sum_{\Ac:\Ac\subseteq [1:K], |\Ac|=p, 1 \in \Ac}S_{\Ac}^{c_1}Q/K \]  (coded and uncoded) intermediate values need to be delivered from Node~$1$ with each intermediate value carrying at most $B$ bits of information.  The condition in \eqref{eq:gencond2}  guarantees that the  scheme design satisfies the communication load constraint for Node~$1$, that is  
\begin{align}
\frac{K \cdot \Hen(x_{1}) }{QB}  \leq  \sum\limits_{p=1}^{K-1}  \frac{K-p}{p}  \sum\limits_{\Ac:\Ac\subseteq [1:K], |\Ac|=p, 1 \in \Ac}S_{\Ac}^{c_1} \leq L_1. \label{eq:finalgen}
\end{align} 
Similarly, the condition in \eqref{eq:gencond2} guarantees that the scheme design satisfies the communication load constraints for other nodes as well.

\begin{figure}[h!]
\centering
\includegraphics[width=13cm]{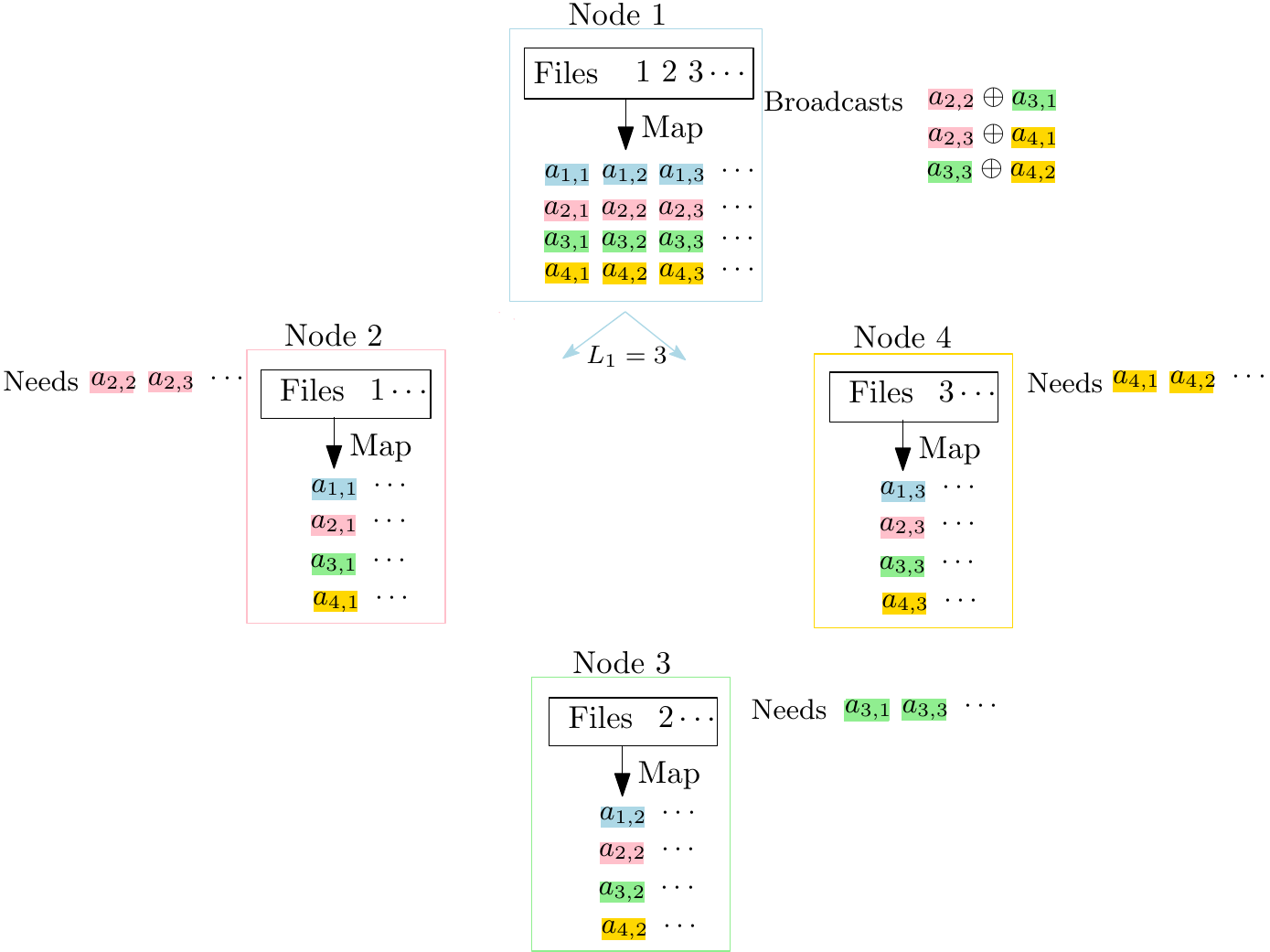} 
\caption{A snippet of an example for \textit{coding strategy} for a $4$-node distributed system to achieve both the  \textit{minimum total computation load} and the  \textit{ minimum worst-case computation load}. For  a group of $ 3$ files, \textit{coding strategy} is utilized during the Shuffle phase. The constraints in \eqref{eq:gencond1} and \eqref{eq:coding} are satisfied with design parameters, $S_{12}^{c_1} = S_{13}^{c_1} = S_{14}^{c_1}=1$.  For file indices $\{1,2,3\}$,  $3$ \emph{coded} intermediate values, i.e., $\{a_{2,2} \oplus a_{3,1}, a_{2,3} \oplus a_{4,1}, a_{3,3} \oplus a_{4,2}\}$, are broadcast from Node~$1$ to the other $3$ nodes. With this design, the communication load constraint for Node~$1$ is also satisfied which can be seen from \eqref{eq:finalgen} with $L_1=3$.}  \label{fig:general_example}
 \end{figure}

The file placement, Map, Shuffle, and Reduce phases of the proposed scheme are explained in Algorithm~\ref{generalKAlgworst}, which takes all the $\sum_{p=1}^{K-1} p  \binom{K}{p} + 1$ parameters (see \eqref{eq:parameters}) satisfying conditions~\eqref{eq:gencond1}-\eqref{eq:gencond2} as inputs. Then, the following worst-case computation load is achievable
{\color{blue}{
\begin{algorithm}
\caption{General Achievable Scheme for Both $\Mt$ and $\Mwt$ with ($Q/K \in \mathbb{N}^+$)}\label{generalKAlgworst}
\begin{algorithmic}[1]
\Procedure{Initialization}{}
\State $z \gets 0$
\For{$p = 1:K-1$}
	\For {\textbf{each} $\Ac: \Ac \subseteq [1:K], |\Ac| = p$}
		\For{$j \in \Ac$}
			\State $\Sc_{\Ac}^{c_j} \gets [z+1:z+S_\Ac^{c_j}]$;  \quad $z \gets z+S_\Ac^{c_j}$
		\EndFor
	\EndFor
\EndFor
\State $\Sc_{\{1\cdots K\}} \gets [N-S_{1\cdots K} + 1:N]$
\EndProcedure
\Procedure{File Assignment}{}
\For{ $k =  1: K $ }
	\State $\Mc_k \gets  (\cup_{p=1}^{K-1} \cup_{\Ac: \Ac \subseteq [1:K], |\Ac| = p, k \in \Ac, } \cup_{j: j \in \Ac} \Sc_\Ac^{c_j}) \cup \Sc_{\{1\cdots K\}}$
	\State  Place all the files indexed by $\Mc_k$ at Node $k$
\EndFor
\EndProcedure
\Procedure{Map phase}{}
\For{ $k =  1: K $ }
	\State  Node $k$ computes \emph{Map} functions and outputs $a_{q,n}$, $q\in [1:Q]$ and $n \in \mathcal{M}_k$
\EndFor
\EndProcedure

\Procedure{Shuffle Phase}{}
\For{$k = 1:K$}
	\State Node~$k$ sends $a_{\Wc_{(k+i)^\dagger},\Sc_{\{k\}}^{c_k}}$ to Node~$(k+i)^\dagger, \  i =1,2,\cdots, K-1$
	\For{$p = 2: K-1$}

		\For{\textbf{each} $\Ac^\prime \subseteq [1:K], |\Ac^\prime| = p, k\notin \Ac^\prime$}

					\State Node~$k$ broadcasts $\bigoplus\limits_{\forall \Ac: \Ac \subseteq [1:K], k \in \Ac, |\Ac| = p, |\Ac \cap \Ac^\prime| = p-1} a_{\Wc_{\Ac^\prime \setminus \Ac}, \Sc_{\Ac}^{c_k}}$ to the other nodes
		\EndFor

	\EndFor
\EndFor

\EndProcedure
\Procedure{Reduce phase}{}
\For{ $k =  1: K $ }
	\State Node $k$ computes \emph{Reduce} functions indexed by $\Wc_k$
\EndFor
\EndProcedure
\end{algorithmic}
\end{algorithm}
}}

\begin{align} \label{eq:generalworsteq}
\Mwtn  &=    \min \Big\{ \max_{i \in \{1 \cdots K\}} \big\{ S_{1\cdots K} + \sum_{p =1}^{K-1} \sum_{\Ac: \Ac \subseteq [1:K], |\Ac| = p, i \in \Ac}    \sum_{j: j \in \Ac}S_{\Ac}^{c_j} \big\} \Big\},\\
\text{s.t.}  \ \ &  \text{conditions~\eqref{eq:gencond1}-\eqref{eq:gencond2}}, \non  
\end{align}
and the following total computation load is achievable
\begin{align} \label{eq:generaltotaleq}
\Mtn  &=    \min  \Big\{\sum_{i=1}^K \ \big(S_{1\cdots K} + \sum_{p =1}^{K-1} \sum_{\Ac: \Ac \subseteq [1:K], |\Ac| = p, i \in \Ac}    \sum_{j: j \in \Ac}S_{\Ac}^{c_j}\big)\Big\}, \\
\text{s.t.}  \ \ &  \text{conditions~\eqref{eq:gencond1}-\eqref{eq:gencond2}}. \non  
\end{align}

We analyze the scheme for the following three cases. At first we look at the system with $K=2$ nodes in Theorem~\ref{thm:K2} and then we look into distributed computing system with $K$ nodes in Propositions~\ref{prop:KsmallL} and \ref{prop:KlargeL}. 
\subsection{Parameter design for Theorem~\ref{thm:K2}} \label{case:0}
In this system with $K=2$, the parameters are set as below. Note that for the following parameter design, WLOG, we assume $L_1 \leq L_2$.
\begin{align*} 
&S_{1}^{c_1} = \min\Big\{L_1, \ceil*{\frac{N}{2}}\Big\}, \\
&S_{2}^{c_2} = \min \Bigg\{L_2, N - \min\bigg\{L_1, \ceil*{\frac{N}{2}}\bigg\}\Bigg\}, \\
&S_{12} = \big(N- L\big)^+ . 
\end{align*}
One can easily verify that \eqref{eq:gencond1} and \eqref{eq:gencond2} hold true for these parameters. Then, the worst-case computation load from \eqref{eq:generalworsteq} is given by
\begin{align} \label{eq:thm1}
\Mwtn = N- \min \{L_1, L_2, N - \ceil*{N/2}\},
\end{align}
and the total computation load from \eqref{eq:generaltotaleq} is given by
\begin{align}\label{eq:thm2}
\Mtn = \max \{N, 2N-L\}.
\end{align}
Equations in \eqref{eq:thm1} and \eqref{eq:thm2} reveal that both the worst-case computation load and the total computation load can be achieved via the same scheme for a distributed computing system with $K=2$ nodes. Both of them turn out to be optimal (see Theorem~\ref{thm:K2}).

\subsection{Parameter design for Proposition~\ref{prop:KsmallL}} \label{case:1}
For  the  case of \[{L}\leq \frac{N}{K-1}, \]
we set the parameters as
\begin{align*} 
&S_{\Ac}^{c_j}=0,  \  \ \Ac \subseteq [1:K], |\Ac| = p, p = 1,2,\cdots,K-2, \forall j \in \Ac,\\
&S_{\Ac}^{c_j} = L_j,  \  \ \Ac \subseteq [1:K], |\Ac| =K-1, \forall j \in \Ac, \\
&S_{1\cdots K} = N - (K - 1)L.
\end{align*}
Putting these values in \eqref{eq:gencond1}, we have
\begin{align}
\binom{K-1}{K-2} \sum_{j = 1}^{K} L_j + N - (K-1)L &= N. \label{eq:generalcondproof1}
\end{align}
Again, setting the parameters in \eqref{eq:gencond2}, we get,
\begin{align}
\frac{1}{K-1} \binom{K-1}{K-2}L_i &\leq L_i,  \quad \quad \quad \forall i \in [1:K]. \label{eq:generalcondproof2}
\end{align}
From \eqref{eq:generalcondproof1} and \eqref{eq:generalcondproof2}, we prove that the choices of the above parameters satisfy the constraints in \eqref{eq:gencond1} and \eqref{eq:gencond2}. Then, the worst-case computation load from \eqref{eq:generalworsteq} is given by 
\begin{align}
\Mwtn = N - L + \max_k \{L_k\},
 \end{align}
 and the total computation load from \eqref{eq:generaltotaleq} is given by 
 \begin{align}
\Mtn = KN - (K-1)L,
 \end{align}
both of which turn out to be optimal (see Proposition~\ref{prop:KsmallL}).

\subsection{Parameter design for Proposition~\ref{prop:KlargeL}} \label{case:2}
Let us consider the case of \[\min_k \{{L}_k\} \geq (K-1) \cdot \ceil*{\frac{N}{K}}.\]
Here, the parameters are set as follows
\begin{align*} 
&S_i^{c_i} =\ceil*{\frac{N}{K}}, \ \   i = 1,2, \cdots,  [N  \text{ mod }  K], \non \\
&S_i^{c_i} =\floor*{\frac{N}{K}}, \ \  i = [N\text{ mod }K]+ 1, \cdots, K-1, K, \\
&S_{\Ac}^{c_j}=0,  \  \ \Ac \subseteq [1:K], |\Ac| = p, \ \ p = 2,3, \cdots , K-1, \forall j \in \Ac, \non \\
&S_{1\cdots K}=0.
\end{align*}
With the above choices of parameters, one can easily verify that the conditions \eqref{eq:gencond1} and \eqref{eq:gencond2} are satisfied. From \eqref{eq:generalworsteq}, the worst-case computation load is given by
\[\Mwtn = \ceil*{N/K}, \]
and from \eqref{eq:generaltotaleq}, the total computation load is given by 
\[\Mtn = N,\]
which turn out to be optimal (see Proposition~\ref{prop:KlargeL}).

\section{Converse \label{sec:converse2}}
 
This section provides the converse proofs of Theorems~\ref{thm:K2} and \ref{thm:K3} and also provides the proof of Lemmas~\ref{lm:converseMwt} and \ref{prop:convbound5}. Note that Lemma~\ref{prop:convbound5} is derived from the proof by contradiction. 
We first provide a lemma that will be used in our proofs.

\begin{lemma} [Cut-Set Bound] \label{lm:converseg}
For a distributed computing system defined in Section~\ref{sec:system}, and for any set $\Sc \subseteq [1:K]$ and $\Sc^c  \defeq [1:K]\setminus \Sc $, we have
\begin{align}
 |\Sc| \cdot \sum_{i\in \Sc} M_i  &\geq  N\cdot |\Sc| -   \sum_{k\in \Sc^{c}} L_k.           \label{eq:conv22}  
\end{align} 
\end{lemma}
\begin{proof}
Lemma~\ref{lm:converseg} holds for a general setting.  We essentially use a ``cut-set'' technique  in the proof of this lemma. 
For notational convenience,  let \[a_{q,:} \defeq \{ a_{q,n} \colon n\in [1:N] \},\] which denotes all the intermediate values required for computing  the Reduce function $q$; let \[a_{\Wc_k,:} \defeq  \{ a_{q,n} \colon  q\in  \Wc_k,  n\in [1:N] \}, \] which represents  all the intermediate values that are required as inputs for computing  all the Reduce functions at Node~$k$; and let \[a_{:,\Mc_k}\defeq  \{ a_{q,n} \colon  q\in  [1:Q],   n\in  \Mc_k\}, \]  which denotes all the intermediate values cached at Node~$k$ after the Map phase, for any $k\in [1:K]$.
For any set $\Sc \subseteq [1:K]$,  we let 
\begin{align}
\Wc_{\Sc}           & \defeq  \cup_{i\in \Sc} \Wc_{i},     \label{eq:def101} \\
\Mc_{\Sc}            & \defeq  \cup_{i\in \Sc} \Mc_{i},     \label{eq:def102} \\
a_{\Wc_{\Sc},:}   & \defeq  \{ a_{q,n} \colon  q\in  \Wc_{\Sc},  n\in [1:N] \},   \\
a_{:, \Mc_{\Sc}}   & \defeq  \{ a_{q,n} \colon  q\in  [1:Q],   n\in  \Mc_{\Sc}\},  \\
x_{\Sc}                &\defeq    \{x_i \colon  i \in \Sc  \}.      \label{eq:def105}
\end{align}
Recall that $x_{k}$ is the message sent from Node~$k$ and is a deterministic function of the intermediate values $a_{\mathcal{W}_k^c,\Mc_k}$, i.e., $x_{k} = f_{k}(a_{\mathcal{W}_k^c,\Mc_k})$. 
Also recall that $x_{k}$ needs to be communicated under the following communication load constraint:
\begin{align}
 \Hen(x_{k}) \leq \frac{QBL_{k}}{K},  \quad  k=1,2,\cdots, K      \label{eq:chain11}
\end{align}
(see \eqref{eq:comc}).
In this proof we use a ``cut-set''  technique. Let us consider an arbitrary  ``cut'' to divide $K$ nodes into two groups. Let $\Sc$ and  $\Sc^c$ denote the sets of node indices of  the first and the second groups, respectively, for $\Sc \subseteq [1:K]$ and $\Sc^c  = [1:K] \setminus \Sc$.

We first argue that the following equality should hold:  
\begin{align}\label{eq:conden}
\Hen(a_{\Wc_{\Sc},:} | a_{\Wc_{\Sc},\Mc_{\Sc}}, x_{\Sc^c}, a_{\Wc_{\Sc^c},:}) = 0,
\end{align}
because  $a_{\Wc_{\Sc},:}$ can be reconstructed by the knowledge of $a_{\Wc_{\Sc},\Mc_{\Sc}}$ and $x_{\Sc^c}$. Note that  the intermediate values $a_{\Wc_{\Sc},:}$ are used as inputs for computing  all  the Reduce functions at a group of nodes indexed by set $\Sc$. 
Also note that $x_{\Sc^c}$ represents a set of messages (see \eqref{eq:def105}), which are sent from a group of nodes indexed by set $\Sc^c$. 
Then, with  \eqref{eq:conden} and the chain rule, we directly get 
\begin{align} \label{eq:chain}
\Hen( a_{\Wc_{\Sc},\Mc_{\Sc}}, x_{\Sc^c}, a_{\Wc_{\Sc^c},:}, a_{\Wc_{\Sc},:}) =  \Hen(a_{\Wc_{\Sc},\Mc_{\Sc}}, x_{\Sc^c}, a_{\Wc_{\Sc^c},:}) .
\end{align}
For the left-hand expression of (\ref{eq:chain}), we have 
\begin{align}
\Hen( a_{\Wc_{\Sc},\Mc_{\Sc}}, x_{\Sc^c}, a_{\Wc_{\Sc^c},:}, a_{\Wc_{\Sc},:})  &= \Hen( a_{\Wc_{\Sc^c},:}, a_{\Wc_{\Sc},:} ) \label{eq:rtst1}\\
&=QNB  \label{eq:2nb}
\end{align}
where  \eqref{eq:rtst1} follows from the fact that  $(a_{\Wc_{\Sc},\Mc_{\Sc}}, x_{\Sc^c})$ can be reconstructed from  $(a_{\Wc_{\Sc^c},:}, a_{\Wc_{\Sc},:} )$;
\eqref{eq:2nb} stems from our definition on the size of each intermediate value. 
Focusing on the right-hand side of  \eqref{eq:chain}, we have
\begin{align}
\Hen(a_{\Wc_{\Sc},\Mc_{\Sc}}, x_{\Sc^c}, a_{\Wc_{\Sc^c},:}) &\leq \underbrace{\Hen(a_{\Wc_{\Sc},\Mc_{\Sc}})}_{ = B \cdot |\Wc_{\Sc} | \cdot |\Mc_{\Sc}| }   + \underbrace{\Hen(x_{\Sc^c})}_{\leq \frac{Q B}{K} \cdot \sum_{k \in\Sc^c } L_k} +  \underbrace{\Hen(a_{\Wc_{\Sc^c},:})}_{ = BN \cdot |\Wc_{\Sc^c}| }  \non \\
  &\leq  B \cdot |\Wc_{\Sc} | \cdot |\Mc_{\Sc}|  +  \frac{Q B}{K} \cdot \sum_{k \in\Sc^c } L_k  +     BN \cdot |\Wc_{\Sc^c}|   \label{eq:ML11}  \\ 
  &\leq \frac{QB}{K} \Big( | \Sc| \cdot  \sum_{i\in \Sc } M_i  + \sum_{k \in\Sc^c } L_{k}  +  N \cdot (K- |\Sc|) \Big)   \label{eq:ML}
\end{align}
where \eqref{eq:ML11} stems from \eqref{eq:chain11} and our definition on the size of each intermediate value;
\eqref{eq:ML}  uses the fact that $ |\Mc_{\Sc}|  \leq  \sum_{i\in \Sc } M_i $   (see  \eqref{eq:def102}),  that $|\Wc_{\Sc} | = Q \cdot | \Sc| /K$ (see \eqref{eq:funcsym} and \eqref{eq:def101}), and that  $|\Wc_{\Sc^c} | = Q \cdot (K- | \Sc|) /K$.
Finally, combining \eqref{eq:chain}, \eqref{eq:2nb} and \eqref{eq:ML} gives the following bound
\begin{align*}
QNB \leq \frac{QB}{K} \Big( | \Sc| \cdot  \sum_{i\in \Sc } M_i  + \sum_{k \in\Sc^c } L_{k}  +  N \cdot (K- |\Sc|) \Big), 
\end{align*}
which can be easily simplified to bound \eqref{eq:conv22}, that is,  $N\cdot |\Sc|  \leq   |\Sc| \cdot \sum_{i\in \Sc} M_i  +  \sum_{k\in \Sc^{c}} L_k$.
\end{proof}

\subsection{Converse proof of  Theorem~\ref{thm:K2} \label{sec:ThmTK2}}

For a two-node ($K=2$) distributed computing system,  Lemma~\ref{lm:converseg} gives the bounds 
\begin{align}
 M_1  &\geq  N - L_2,      \label{eq:conv101}  \\
  M_2  &\geq  N - L_1,  \label{eq:conv102}  \\ 
   M_2 + M_1  &\geq  N, \label{eq:conv103}  
\end{align}
by setting $\Sc = \{1\} $,  $\Sc = \{2\} $, and $\Sc = \{1, 2\}$, respectively.  
Therefore, by combining   bounds~\eqref{eq:conv101}-\eqref{eq:conv103}, we have
\begin{align} \label{eq:theorem1a}
\Mtn &\ge \max\{N, 2N-L\}.
\end{align}
From \eqref{eq:conv101} and \eqref{eq:conv102} we have 
\begin{align}
  \Mwtn   &\geq  N - \min \{L_1, L_2\} .   \label{eq:convL1L2}  
\end{align}
Furthermore,  bound~\eqref{eq:conv103} also implies that 
\begin{align}
   \Mwtn   &\geq  \ceil*{N/2}.  \label{eq:conv104}  
\end{align}
Therefore, combining   bounds~\eqref{eq:convL1L2} and \eqref{eq:conv104}, we have
\begin{align}\label{eq:theorem1b}
\Mwtn  &\ge N - \min \{L_1, L_2, N - \ceil*{N/2}\}.
\end{align}
The bounds \eqref{eq:theorem1a} and \eqref{eq:theorem1b} give the converse proof for Theorem \ref{thm:K2}.

\subsection{Converse proof of  Theorem~\ref{thm:K3} \label{sec:ThmTK3}}

For a distributed computing system with $K=3$,  Lemma~\ref{lm:converseg} gives the following bounds 
\begin{align}
 M_1  &\geq  N - (L_2 + L_3),      \label{eq:corr31}  \\
  M_2  &\geq  N - (L_1 + L_3),  \label{eq:corr32}  \\ 
   M_3 &\geq  N - (L_1 + L_2),    \label{eq:corr33}  
\end{align}
by setting $\Sc = \{1\}$, $\{2\}$ and $\{3\}$, respectively. Then, by combining bounds \eqref{eq:corr31}-\eqref{eq:corr33} we have
\begin{align}
\Mtn &\geq 3N - 2L \label{eq:sum3a}.
\end{align}
Furthermore, by setting  $\Sc = \{1,2,3\} $, Lemma~\ref{lm:converseMwt} gives the following bound
\begin{align}
\Mtn &\geq  N.       \label{eq:corr311}  
\end{align}
Finally, for a distributed computing system with $K=3$, the total computation load is bounded by 
\begin{align}\label{eq:otherboundproof}
\Mtn \geq \ceil*{\frac{7N-2L}{3}}
\end{align}
by following the result in \cite[Lemma~1]{LMYA:17} (also see \cite[Theorem~1]{KWA:17}).
Bound~\eqref{eq:otherboundproof} also uses the integer property of $M_k$, $\forall k \in [1:K]$.
Therefore,  with bounds~\eqref{eq:sum3a}, \eqref{eq:corr311} and \eqref{eq:otherboundproof}, we complete the converse proof of Theorem~\ref{thm:K3}.

\subsection{Converse proof of  Proposition~\ref{prop:KsmallL} \label{sec:prop3small}}
For a distributed computing system with $K$ nodes, Lemma~\ref{lm:converseg} gives the following bounds when $|\Sc| = 1$
\begin{align} \label{eq:totalKbound1}
M_i \geq N -L + L_i, \quad   i \in [1:K].
\end{align}
Combining all the bounds for $\Sc \subseteq [1:K], |\Sc| = 1$, we have
\begin{align} \label{eq:worstKbound1}
\Mtn = \sum_{i=1}^{K} M_i \geq KN - (K-1)L.
\end{align}
Furthermore, bound \eqref{eq:totalKbound1} also implies that
\begin{align} \label{eq:worstKbound2}
\Mwtn \geq N-L+\max_{k} \{L_k\}.
\end{align}
Hence, the bounds \eqref{eq:worstKbound1} and \eqref{eq:worstKbound2} provide the converse for Proposition~\ref{prop:KsmallL}.

\subsection{Converse proof of  Proposition~\ref{prop:KlargeL} \label{sec:prop3large}} 
 For a distributed computing system with $K$ nodes, Lemma~\ref{lm:converseg} gives the following bounds
\begin{align}\label{eq:totalKbound}
\Mtn = \sum_{i = 1}^{K} M_i  &\geq  N,
\end{align}
by setting $\Sc = [1:K]$ with $|\Sc| = K$. Furthermore, bound \eqref{eq:totalKbound} also implies that
\begin{align} \label{eq:worstKbound}
\Mwtn \geq \ceil*{N/K}.
\end{align}
The bounds \eqref{eq:totalKbound} and \eqref{eq:worstKbound} give the converse proof for Proposition~\ref{prop:KlargeL}.

\subsection{Proof of  Lemma~\ref{lm:converseMwt} \label{sec:lemWK3} }
For a distributed computing system with $K=3$,  Lemma~\ref{lm:converseg} gives the bounds
\begin{align}
 M_1 + M_2  &\geq  N - \frac{L_3}{2},      \label{eq:corr301}  \\
  M_2 + M_3 &\geq  N - \frac{L_1}{2},  \label{eq:corr302}  \\ 
   M_3 + M_1 &\geq  N - \frac{L_2}{2}, \label{eq:corr303}  
\end{align}
by setting $\Sc =  \{1,2\}, \{2,3\}$ and $\{3,1\}$, respectively. Then, the following bound is directly from the combination of bounds~\eqref{eq:corr301}-\eqref{eq:corr303}:
\begin{align}
 \Mwtn \geq \ceil*{\frac{N}{2} - \frac{\min_k \{L_k\}}{4}}. \label{eq:worst3b}
\end{align}
The above bound also uses the integer property of $M_k$, $\forall k \in [1:K]$.
Furthermore,  by combining bounds~\eqref{eq:corr31}-\eqref{eq:corr33} it gives 
\begin{align}
 \Mwtn &\geq N - \min \{ L_2 +L_3, L_1+L_3,L_1+L_2\}. \label{eq:worst3a}
\end{align}
Finally, bounds~\eqref{eq:corr311} and  \eqref{eq:otherboundproof} imply  the following bounds
\begin{align}
 \Mwtn  \geq \ceil*{N/3}, \label{eq:sum3c}
\end{align}
and 
\begin{align}
 \Mwtn  \geq   \ceil*{\frac{\ceil*{\frac{7N-2L}{3}}}{3}}, \label{eq:worst3c343}
\end{align}
respectively. 
Therefore,  with bounds~\eqref{eq:worst3b}-\eqref{eq:worst3c343}, we complete the  proof of Lemma~\ref{lm:converseMwt}.

\section{Proof of  Lemma~\ref{prop:convbound5} (proof by contradiction) \label{sec:lemWK3Con}  }

In this section we will prove Lemma~\ref{prop:convbound5} using the proof by contradiction, for a  three-node $(K=3)$ distributed computing system. Note that for the proof by contradiction, we start by assuming opposite proposition is true and then show that this assumption leads to a contradiction. 
Specifically, for any nonnegative integer $\beta \in \mathbb{N}$ satisfying the  following condition,
\begin{align} 
\beta \leq \ceil[\Big]{\ceil[\Big]{ \Bigl(7N-2 \cdot \sum_{k=1}^3   \min\{L_k,  \ 2(\beta - 1)\} \Bigr)\Big/  3}\Big/3},    \label{eq:betacon} 
\end{align} 
we will prove that 
\begin{align}
 \Mwtn  \geq \beta. \label{eq:toprovebeta}
\end{align}

\subsection{Assume \eqref{eq:toprovebeta} is false} \label{lb:proofbycontra1}
At first, we assume that 
\begin{align}
 \Mwtn  \leq \beta - 1, \label{eq:letcontradict}
\end{align}
where $\beta$ is a nonnegative integer satisfying  the  condition in \eqref{eq:betacon}.

\subsection{Implication of Assumption \ref{lb:proofbycontra1}}
Based on the assumption in \eqref{eq:letcontradict}, it implies that 
\begin{align}
\Hen (x_k) &\leq  \Hen\bigl(  a_{\mathcal{W}_k^c,\Mc_k} \bigr) \label{eq:proofBineq1}  \\
&\leq  M_k  B \cdot \bigl(Q -Q/K\bigr)   \non \\
&\leq  (\beta - 1) \cdot B Q  (K-1)/K \label{eq:proofBineq}
\end{align}
for $k=1,2, 3$, where  \eqref{eq:proofBineq1} follows from the definition of $x_{k} = f_{k}(a_{\mathcal{W}_k^c,\Mc_k})$  (see \eqref{eq:xkdef}) and the identity of $\Hen (f(e)) \leq \Hen(e)$ for a deterministic function $f(e)$; \eqref{eq:proofBineq} uses the assumption in  \eqref{eq:letcontradict}.
Therefore, by combining \eqref{eq:proofBineq} and the communication load constraint in \eqref{eq:comc}, we have
\begin{align}
\Hen (x_k)  \leq \min \Big\{\frac{L_k Q B}{K}, \  \frac{(\beta - 1 )\cdot (K-1) \cdot QB }{K}\Big\}.     \label{eq:cond1123}
\end{align}
By defining a new parameter as 
\begin{align}
{L_k'} \defeq \min \{L_k, \  (\beta - 1)\cdot (K-1)\}, \label{eq:cond5a}
\end{align}
then \eqref{eq:cond1123} implies that we have the following communication load constraint:
\begin{align}
\Hen (x_k)  \leq \frac{{L_k'} QB}{K}     \label{eq:condnewlk}
\end{align}
 for $k=1,2,3$.
Based on \eqref{eq:condnewlk}, we have a new parameter on the total communication load constraint:
\begin{align*}
L'&  \defeq  \sum_{k=1}^3 {L_k'}\\
&=  \sum_{k=1}^3  \min \{L_k,  \  2(\beta - 1)\}
\end{align*}
for this setting with $K=3$. 
\subsection{Using the new parameter  $L'$ on the total communication load constraint}
At this point, by replacing $L$ with $L'$ and  by using the result of bound~\eqref{eq:worst3c343}, it gives 
\begin{align} 
 \Mwtn  \geq   \ceil*{ \ceil*{\frac{7N -2 \cdot \sum_{k=1}^3  \min \{L_k, \  2(\beta - 1)\}}{3}}\Big/3}.   \label{eq:worstbeta22} 
\end{align} 
\subsection{Contradiction}
From \eqref{eq:worstbeta22} and the condition in \eqref{eq:betacon}, it then implies that  
\begin{align} 
 \Mwtn  \geq   \ceil*{ \ceil*{\frac{7N -2 \cdot \sum_{k=1}^3  \min \{L_k, \  2(\beta - 1)\}}{3}}\Big/3} \geq \beta,  \label{eq:worstbeta11} 
\end{align} 
which contradicts with the assumption in \eqref{eq:letcontradict}. 
\subsection{Conclusion}
Therefore, the bound of $ \Mwtn  \ge \beta$ holds true for  any nonnegative integer $\beta \in \mathbb{N}$ satisfying the  condition  in \eqref{eq:betacon}. 
Finally, for  a nonnegative integer $\beta^*$ defined by the following optimization problem
\begin{align} 
\beta^* \!=\! \max  \  & \beta, \non \\
  \text{s.t.}   \   & \beta \leq \ceil[\bigg]{\ceil[\Big]{ \bigl(7N-2 \cdot \sum_{k=1}^3   \min\{L_k, \  2(\beta - 1)\} \bigr)\Big/  3}\Big/3},  \non \\
   & \beta \in \mathbb{N}  \non
\end{align}
($\beta^*$ is also a nonnegative integer satisfying the  condition  in \eqref{eq:betacon}), we conclude from the above argument  that 
\[ \Mwtn  \geq \beta^*,\] 
which completes the proof of  Lemma~\ref{prop:convbound5}.

\section{Proof of the tradeoff between $\Mt$ and $\Mwt$} \label{sec:contradp2}

In this section, we provide the proof of Theorem~\ref{thm:tradeoff}.
For a distributed computing system with $K$ nodes, we have proposed a general scheme in Section~\ref{sec:achiKsame} that indeed can always achieve $\Mt$ and $\Mwt$  at the same time. 
For the setting with $K=3$ we will prove that, if the following condition 
\begin{align} 
\Mt & <  3N- 2\Mwt - \min_{i\neq j}(L_i + L_j) \non
\end{align}
 (see also \eqref{eq:tradeoffeq}) is satisfied, then $\Mt$ and $\Mwt$  cannot be achieved at the same time. 
In this section we will focus on the setting with $K=3$. First, we provide a lemma that will be used in our proof. The proof of this lemma will be provided in Section~\ref{sec:lemMk} later on. Recall that  $S_j \defeq |\Sc_{\{j\}} | \defeq | \Mc_j \setminus \cup_{i \in [1:K]\setminus j  }\Mc_{i}| $ for $j\in [1:K]$ and similar notations are defined in \eqref{eq:Sdefine}.

\begin{lemma}  \label{lm:contrad}
For a  three-node ($K=3$) distributed computing system defined in Section~\ref{sec:system}, we have
\[M_k \ge N- \sum_{i=1, i \neq k}^{3} L_i+   \sum_{j=1, j \neq k}^3 S_j,\] for $k=1,2,3$. 
\end{lemma}

At first, assume that there is a MapReduce scheme that achieves  the minimum worst-case computation load, i.e., 
\begin{align}   \label{eq:MwAssume}
 \Mwtn  = \Mwt.  
\end{align}
From  Lemma~\ref{lm:contrad}, we have one inequality given as 
\begin{align}
S_1 +S_2 \leq M_3+ L_1 + L_2 - N.  \label{eq:cond500a}
\end{align}
Furthermore,  we have
\begin{align}
S_1+ S_2+ S_{12} + M_3 = N,  \label{eq:cond5525}
\end{align}
based on the system requirement that each of the $N$ files should be assigned to at least one node, where $M_3 = S_3 + S_{13}+S_{23}+ S_{123}$ (see the corresponding definition in \eqref{eq:Sdefine}).
Then, combining \eqref{eq:cond500a} and \eqref{eq:cond5525} gives the following bound on $S_{12}$:
\begin{align}
S_{12}&=N-M_3-(S_1+S_2) \non\\
&\ge 2N-2M_3-(L_1 +L_2)  \label{eq:cond88231} \\
&\ge 2N- 2 \Mwt - (L_1 +L_2),  \label{eq:cond8823}
\end{align}
where \eqref{eq:cond88231} follows from \eqref{eq:cond500a};  \eqref{eq:cond8823} is based on the assumption in \eqref{eq:MwAssume}. 
At this point,  we have the following bounds on the total computation load:
\begin{align}
\Mtn  &= \underbrace{S_{1}+S_{12}+S_{13}+ S_{123}}_{M_1} +   \underbrace{S_{2}+ S_{12} + S_{23}+S_{123}}_{M_2} + M_3   \label{eq:Msum332} \\
&=S_{1}+S_{2}+2S_{12}+S_{13}+S_{23}+2S_{123}\!+\!M_3\non\\
&\ge S_{1}+S_{2}+2S_{12}+M_3 \label{eq:cond55a}\\
&= N  + S_{12} \label{eq:cond57a11} \\
&\ge 3N-2 \Mwt - (L_1 + L_2), \label{eq:cond57a}
\end{align}
where \eqref{eq:Msum332} uses the identity of $M_1 = S_{1}+S_{12}+S_{13}+ S_{123}$ and $M_2= S_{2}+ S_{12} + S_{23}+S_{123}$;
\eqref{eq:cond55a} follows from the property of $S_{13}, S_{23}, S_{123} \ge 0$; \eqref{eq:cond57a11} is from  \eqref{eq:cond5525}; and \eqref{eq:cond57a} is from  \eqref{eq:cond8823}. 
Similarly, by following the above steps we also have
\begin{align}
 \Mtn \geq 3N-2 \Mwt - (L_2 + L_3), \label{eq:cond57b}
\end{align}
and
\begin{align}
\Mtn \geq 3N-2 \Mwt - (L_1 + L_3), \label{eq:cond57c}
\end{align}
which, together with \eqref{eq:cond57a}, give the following conclusion
\begin{align}
\Mtn  \ge  3N- 2\Mwt - \min_{i\neq j}(L_i + L_j), 
\end{align}
if the condition in \eqref{eq:MwAssume} is satisfied.  
Based on this conclusion, it implies that when $\Mt  <  3N- 2\Mwt - \min_{i\neq j}(L_i + L_j) $ is satisfied, $\Mt$ and $\Mwt$ cannot be achieved at the same time. It then proves Theorem~\ref{thm:tradeoff}.

\subsection{Proof of Lemma~\ref{lm:contrad}}  \label{sec:lemMk}

Let us now prove Lemma~\ref{lm:contrad} that has been used in our proof. 
Beginning with the communication load constraint in \eqref{eq:comc}, we have  
\begin{align}
\frac{ (L_1 + L_2)\cdot QB}{K} &\ge \Hen(x_1)  + \Hen(x_2) \label{eq:comcx1x2}\\
 &\ge \Hen(x_1, x_2) \label{eq:lema}\\
&\ge \Hen(x_1, x_2 | a_{:,\Mc_3}) \label{eq:lemb}\\
&= \Hen(x_1, x_2 | a_{:,\Mc_3}) + \underbrace{\Hen(a_{\Wc_3,:}|x_1, x_2, a_{:, \Mc_3}}_{=0}) \label{eq:lemc}\\
&= \Hen(x_1, x_2,a_{\Wc_3,:}|a_{:,\Mc_3}) \label{eq:lemd}\\
&= \underbrace{\Hen(a_{\Wc_3,:}|a_{:,\Mc_3})}_{= |\Wc_3|\cdot N\cdot B - |\Wc_3|\cdot|\Mc_3|\cdot B} +\underbrace{\Hen(x_1, x_2| a_{:,\Mc_3}, a_{\Wc_3,:})}_{\ge (S_1 + S_2)\cdot B\cdot Q/K} \label{eq:leme}\\
&\ge \frac{(N-M_3)\cdot Q B}{K} + \frac{(S_1 + S_2)\cdot Q B}{K},  \label{eq:lemg}
\end{align}
where  \eqref{eq:comcx1x2}  stems from \eqref{eq:comc}; 
\eqref{eq:lema} uses the identity that $\Hen(x_1)+ \Hen(x_2) \geq \Hen(x_1, x_2)$; 
\eqref{eq:lemb} is from the fact that conditioning reduces entropy; \eqref{eq:lemc} holds true because $a_{\Wc_3,:}$ can be reconstructed from $(x_1, x_2, a_{:, \Mc_3})$;  \eqref{eq:lemd} and \eqref{eq:leme} result  from the chain rule; 
\eqref{eq:lemg} follows from the fact that $\Hen(a_{\Wc_3,:} | a_{:,\Mc_3})= |\Wc_3|\cdot N B - |\Wc_3|\cdot M_3  B = (N-M_3)\cdot Q B/ K$ and Lemma~\ref{lm:x1x2} (see below).
At this point, from \eqref{eq:lemg} we have 
\begin{align}
L_1 + L_2 &\geq N - M_3 + (S_1 + S_2) \label{eq:lemg0}.
\end{align}
Similarly, by following the above steps we also have two bounds given as:
\begin{align}
 L_2 + L_3 &\geq N - M_1 + (S_2 + S_3),   \label{eq:lemg1} \\
 L_1 + L_3 &\geq N - M_2 + (S_1 + S_3),   \label{eq:lemg2}
\end{align}
which, together with \eqref{eq:lemg0}, complete the proof of Lemma~\eqref{lm:contrad}, that is,
\begin{align}
M_k \ge N- \sum_{i=1, i \neq k}^{3} L_i+   \sum_{j=1, j \neq k}^3 S_j, \quad  k=1,2,3. 
\end{align}

\vspace{5pt}

\begin{lemma}  \label{lm:x1x2}
For a  three-node ($K=3$) distributed computing system defined in Section~\ref{sec:system},  the following inequalities hold true:
\begin{align} 
\Hen(x_1, x_2| a_{:,\Mc_3}, a_{\Wc_3,:}) &\geq \frac{(S_1 + S_2)\cdot QB}{K},  \label{eq:proofLMb}\\
\Hen(x_1, x_3| a_{:,\Mc_2}, a_{\Wc_2,:}) &\geq \frac{(S_1 + S_3)\cdot QB}{K}, \label{eq:proofLMb13}   \\
\Hen(x_2, x_3| a_{:,\Mc_1}, a_{\Wc_1,:}) &\geq \frac{(S_2 + S_3)\cdot QB}{K}. \label{eq:proofLMb23}  
\end{align} 
\end{lemma}

\begin{proof}
We will focus on the proof of bound~\eqref{eq:proofLMb}, as the proofs of bounds~\eqref{eq:proofLMb}, \eqref{eq:proofLMb13} and \eqref{eq:proofLMb23} are similar.  
We first argue that the following equations  are true given the system constraints:
\begin{align}
\Hen(a_{\Wc_2, :} | x_1, x_2, a_{:,\Mc_3}, a_{\Wc_3,:}, a_{:, \Sc_1^c}) &= 0, \label{eq:Hconstraint1}\\
\Hen(a_{\Wc_1, :} | x_2, x_1, a_{:,\Mc_3}, a_{\Wc_3,:}, a_{:, \Sc_2^c}) &= 0.  \label{eq:Hconstraint2}
\end{align}
The two equations result from the fact that  $a_{\Wc_2,:}$ and $a_{\Wc_1,:}$ can be recovered from the information of  $(x_1, x_2, a_{:,\Mc_3}, a_{\Wc_3,:},  a_{:, \Sc_1^c})$ and $(x_1, x_2, a_{:,\Mc_3}, a_{\Wc_3,:}, a_{:, \Sc_2^c})$, respectively.

Let us first focus  on the equation in \eqref{eq:Hconstraint1}. With the use of chain rule,  \eqref{eq:Hconstraint1} implies that 
\begin{align}
\Hen(x_1, x_2, a_{:,\Mc_3}, a_{\Wc_3,:}, a_{:, \Sc_1^c}) = \Hen(x_1, x_2, a_{:, \Mc_3}, a_{\Wc_3, :}, a_{:, \Sc_1^c}, a_{\Wc_2, :}).  \label{eq:lmproof0}
\end{align}
The right-hand side of \eqref{eq:lmproof0} can be rewritten by using the chain rule: 
\begin{align}
\Hen(x_1, x_2, a_{:, \Mc_3}, a_{\Wc_3, :}, a_{:, \Sc_1^c}, a_{\Wc_2, :}) =  \Hen(x_2, a_{:,\Mc_3}, a_{\Wc_3,:}, a_{:,\Sc_1^c}) +  \Hen(x_1, a_{\Wc_2, :} |x_2, a_{:, \Mc_3}, a_{\Wc_3,:}, a_{:,\Sc_1^c}) .   \label{eq:lmproof2a1}
\end{align}
Similarly, the left-hand side of \eqref{eq:lmproof0} can be rewritten by using the chain rule: 
\begin{align}
 \Hen(x_1, x_2, a_{:,\Mc_3}, a_{\Wc_3,:}, a_{:, \Sc_1^c})  =  \Hen(x_2, a_{:,\Mc_3}, a_{\Wc_3,:}, a_{:,\Sc_1^c}) +
\Hen(x_1|x_2, a_{:, \Mc_3}, a_{\Wc_3,:}, a_{:,\Sc_1^c}) \label{eq:lmproof1}.
\end{align}
Then, by plugging  \eqref{eq:lmproof2a1} and \eqref{eq:lmproof1} into \eqref{eq:lmproof0}, we have
\begin{align}
 \Hen(x_1|x_2, a_{:, \Mc_3}, a_{\Wc_3,:}, a_{:,\Sc_1^c}) = \Hen(x_1, a_{\Wc_2, :} |x_2, a_{:, \Mc_3}, a_{\Wc_3,:}, a_{:,\Sc_1^c})   \label{eq:lmproof111}.
\end{align}
From \eqref{eq:lmproof111} we further have:  
\begin{align}
 \Hen(x_1|x_2, a_{:, \Mc_3}, a_{\Wc_3,:}, a_{:,\Sc_1^c}) &= \Hen(x_1, a_{\Wc_2, :} |x_2, a_{:, \Mc_3}, a_{\Wc_3,:}, a_{:,\Sc_1^c})     \non\\
 &= \Hen(a_{\Wc_2,:}| x_2, a_{:,\Mc_3}, a_{\Wc_3,:}, a_{:,\Sc_1^c}) + \underbrace{\Hen(x_1 | x_2, a_{:,\Mc_3}, a_{\Wc_3,:}, a_{:, \Sc_1^c}, a_{\Wc_2, :})}_{\ge 0}  \label{eq:lmproof2}\\
&\geq \Hen(a_{\Wc_2,:}| x_2, a_{:,\Mc_3}, a_{\Wc_3,:}, a_{:,\Sc_1^c}) \label{eq:lmproof3}\\
&= \Hen(a_{\Wc_2,:}|a_{:,\Sc_1^c}, a_{\Wc_3,:}) \label{eq:lmproof4}\\
&= \frac{QB\cdot (N-|\Sc_1^c|)}{K} \label{eq:lmproof5}\\
&= \frac{QB S_1}{K}, \label{eq:lmproof6}
\end{align}
where \eqref{eq:lmproof2} is from chain rule; 
\eqref{eq:lmproof3} stems from the property that entropy is always nonnegative; 
\eqref{eq:lmproof4} is due to the fact that both $x_2$ and $a_{:,\Mc_3}$ can be recovered by the information of $a_{:,\Sc_1^c}$; \eqref{eq:lmproof5} is from the fact that $\Hen(a_{\Wc_2,:}|a_{:, \Sc_1^c}, a_{\Wc_3,:}) = |\Wc_2| \cdot NB  -   |\Wc_2| \cdot |\Sc_1^c|  \cdot B = QB\cdot (N-|\Sc_1^c|)/K$; \eqref{eq:lmproof6} uses the fact that $S_1= N-|\Sc_1^c|$.
Similarly, focusing on  the equation in \eqref{eq:Hconstraint2} and using the above steps, we also have 
\begin{align}
\Hen(x_2|x_1, a_{:, \Mc_3}, a_{\Wc_3,:}, a_{:,\Sc_2^c}) &\geq  \frac{QB S_2}{K} \label{eq:lmproof6a}.
\end{align}

Finally, by using the results in \eqref{eq:lmproof6} and  \eqref{eq:lmproof6a}, we prove bound~\eqref{eq:proofLMb} as 
\begin{align}
\Hen(x_1, x_2| a_{:,\Mc_3}, a_{\Wc_3,:})  &= \Hen(x_1|x_2,a_{:,\Mc_3}, a_{\Wc_3,:}) + \Hen(x_2|a_{:,\Mc_3}, a_{\Wc_3,:})\label{eq:prooflm101}\\
&\ge \Hen(x_1|x_2, a_{:, \Mc_3}, a_{\Wc_3, :}, a_{:, \Sc_1^c}) + \Hen(x_2|x_1, a_{:, \Mc_3}, a_{\Wc_3, :}, a_{:, \Sc_1^c}) \label{eq:prooflm102}\\
&\geq  \frac{QB S_1}{K} + \frac{QB S_2}{K} \label{eq:prooflm103}\\
&= \frac{(S_1 +S_2) \cdot Q B}{K},
\end{align}
where \eqref{eq:prooflm101} follows from chain rule; \eqref{eq:prooflm102} is due to the fact that conditioning reduces entropy; \eqref{eq:prooflm103} is from \eqref{eq:lmproof6} and \eqref{eq:lmproof6a}. 
At this point, we complete the proof of bound~\eqref{eq:proofLMb}. By using the similar steps and interchanging the roles of nodes, one can also prove bounds~\eqref{eq:proofLMb13} and \eqref{eq:proofLMb23} and complete the whole proof.
\end{proof}

\section{Conclusion}   \label{sec:conclusion}

For the distributed computing systems with heterogeneous communication load constraints, we provided the \emph{information-theoretical} characterization of the   \emph{minimum total computation load} and the \emph{minimum worst-case computation load} for some cases.  
In this setting, we showed that for some cases there is a tradeoff between the minimum total computation load  and the minimum worst-case computation load, in the sense that both cannot be achieved at the same time. 
We also showed that  in some instances, proof by contradiction is a very powerful approach to derive the optimal converse bound. 
Finally, we identified two extreme regimes in which the scheme with coding  and the scheme without coding are optimal, respectively.



\end{document}